\documentclass[11pt]{article}

\usepackage{hyperref}
\usepackage{microtype,amssymb,latexsym,graphicx, xspace, url, xcolor, subfigure}
\usepackage{amsmath,amsthm}
\usepackage[margin=1.25in]{geometry}
\usepackage{mathpazo}
\usepackage{cite}
\usepackage{comment}
\usepackage{enumerate}

\makeatletter

\long\def\@makecaption#1#2{
   \vskip 10pt
   \setbox\@tempboxa\hbox{{\footnotesize {\bf #1.} #2}}
   \ifdim \wd\@tempboxa >\hsize         % IF longer than one line:
       {\footnotesize {\bf #1.} #2\par}% THEN set as ordinary paragraph.
     \else                              %   ELSE  center.
       \hbox to\hsize{\hfil\box\@tempboxa\hfil}
   \fi}
\makeatother

\long\def\oldtext#1{}

\def\reals{{\mathbb R}}

\def\adj{\mathcal{G}}
\def\HDAG{\mathcal{H}}

\def\cell{\tau}
\def\ceil#1{\lceil #1 \rceil}
\def\floor#1{\lfloor #1 \rfloor}

\def\eps{{\varepsilon}}
\def\bd{{\partial}}
\def\A{\mathcal{A}}

\def\C{\mathcal{C}}

\def\F{\mathcal{F}}

\def\K{\mathcal{K}}

\def\R{\mathcal{R}}
\def\S{\mathcal{S}}

\def\U{\mathcal{U}}

\def\EE{\mathsf{E}}

\def\XX{\mathsf{X}}

\def\VeD#1{{#1}^{\mbox{\raisebox{4pt}{\tiny$\|$}}}}

\def\visib{\varphi}

\def\VDsubstruc{\VeD{\C}}
\def\x24{{4}}
\def\subklevel#1{\C^{{\le}#1}}

\def\ASet{\Xi}
\def\CSet{\Phi}

\def\defn{\mathcal{D}}

\def\cross{\mathrm{cr}}

\def\etal{\textsl{et~al.}}
\def\M{\mathsf{M}}

\newtheorem{theorem}{Theorem}[section]
\newtheorem{lemma}[theorem]{Lemma}

\theoremstyle{remark}

\def\subsecref#1{Section~\ref{subsec:#1}}
\def\secref#1{Section~\ref{sec:#1}}
\def\thmref#1{Theorem~\ref{thm:#1}}
\def\lemref#1{Lemma~\ref{lem:#1}}
\def\eqref#1{(\ref{eq:#1})}

\begin{document}

%\begin{titlepage}

\title{Nearly-Tight Bounds for Vertical Decomposition \\ in Three and Four Dimensions\thanks{%
    A preliminary version of this paper appeared in the \emph{Proc. 35th ACM-SIAM Sympos. Discrete Algorithms}, 2024, 150--170.
    We note that the proof given in this proceeding version has a technical gap, so we present a different and more involved approach to obtain the same result. \newline
  Work by Pankaj Agarwal has been partially supported by NSF grants IIS-18-14493, CCF-20-07556, and CCF-22-23870, and by the Binational Science Foundation Grant 2022131.
  Work by Esther Ezra has been partially supported by Israel Science Foundation Grants 800/22 and  824/17, and by the Binational Science Foundation Grant 2022131.
  Work by Micha Sharir has been partially supported by Israel Science Foundation Grants 260/18 and 495/23.}}

\author{Pankaj K. Agarwal\thanks{
    Department of Computer Science, Duke University, Durham, NC 27708, USA;
    {\sf pankaj@cs.duke.edu,
    https://orcid.org/0000-0002-9439-181X}}
  \and
    Esther Ezra\thanks{
    School of Computer Science, Bar Ilan University, Ramat Gan, Israel;
    {\sf ezraest@cs.biu.ac.il,
      https://orcid.org/0000-0001-8133-1335}}
  \and
    Micha Sharir\thanks{
    School of Computer Science, Tel Aviv University, Tel Aviv, Israel;
    {\sf michas@tauex.tau.ac.il,
      http://orcid.org/0000-0002-2541-3763}}
}
\date{}

\maketitle

%\fancyfoot[R]{\scriptsize{Copyright \textcopyright\ 2024 by SIAM\\ Unauthorized reproduction of this article is prohibited}}

\begin{abstract}
 Vertical decomposition is a widely used general technique for decomposing the cells of arrangements of 
  semi-algebraic sets in $\reals^d$ into constant-complexity subcells. In this paper, we settle in the affirmative a 
  few long-standing open problems involving the vertical decomposition of substructures of 
  arrangements for $d=3,4$. For example, we obtain sharp bounds on the complexity of the vertical 
  decomposition of  the complement of the union of a set of semi-algebraic regions of constant 
  complexity in $\reals^3$, and of the minimization diagram of a set of trivariate functions.
  These results lead to efficient algorithms for a variety of problems 
  involving vertical decompositions, including algorithms for constructing the decompositions 
  themselves and for constructing $(1/r)$-cuttings of substructures of arrangements.
  They also lead to a data structure for answering point-enclosure 
  queries amid semi-algebraic sets in $\reals^3$ and $\reals^4$.
\end{abstract}

%%\end{titlepage}

%--------------------------------------------------------
\section{Introduction}
\label{sec:intro}

%\smallskip

Let $\S$ be a family of $n$ semi-algebraic sets\footnote{%
  Roughly speaking, a semi-algebraic set in $\reals^d$ is the set of points in $\reals^d$ that 
  satisfy a Boolean formula over a set of polynomial inequalities; the complexity of a semi-algebraic 
  set is the number of polynomials defining the set and their maximum degree.
  See \cite{BPR} for formal definitions of a semi-algebraic set and its complexity.}
of constant complexity in $\reals^d$. 
The \emph{arrangement} of $\S$, denoted by $\A(\S)$, is the decomposition of $\reals^d$ into 
maximal connected relatively open cells of all dimensions, so that all points within a cell 
lie in the relative interior of the same subfamily of sets of $\S$, and on the boundary of 
another such subfamily. Because of their wide range of applications, arrangements 
of semi-algebraic sets have been extensively studied; see~\cite{AS00,BPR}. The combinatorial
complexity of a cell in $\A(\S)$ can be large, and its topology can be complex~\cite{AS00}, 
so a fundamental problem in the area of arrangements, for both combinatorial and algorithmic applications,
is to decompose a cell of $\A(\S)$ into constant-complexity subcells, each homeomorphic to 
a ball of the appropriate dimension. In some applications, we wish to decompose all cells 
of $\A(\S)$ while in others only a subset of cells of $\A(\S)$.

Vertical decomposition is a popular general technique (and perhaps the only general-purpose 
technique) for constructing such a decomposition. Roughly speaking, vertical decomposition 
partitions a cell of $\A(\S)$ by recursing on the dimension. At each dimension, it partitions
a cell into vertical ``pseudo-prisms'' (prisms for short), each of which has a unique pair 
of floor and ceiling facets, but can otherwise have large complexity, until the recursion 
terminates. At the end of recursion, we get a collection of prisms, each of which consists 
of at most $2d$ facets, and is homeomorphic to a ball of the respective dimension. See 
\subsecref{VD} and~\cite{CEGS-91,Koltun-04a,SA} for a formal definition and for more details.
%recurses on the dimension $d$. Let $C$ be a cell of $\A(\S)$.
%For $d=2$, the vertical decomposition of $C$ is obtained by erecting a $y$-vertical segment up and down 
%from each vertex of $C$ and from each point of vertical tangency on the boundary of $C$,
%and extending these segments till they hit another edge of $C$, or else all the way to infinity. This results in a 
%decomposition of $C$ into vertical \emph{pseudo-trapezoids} (trapezoids, for short).
%For $d=3$, we first erect, upward and downward, $z$-vertical \emph{curtains} from each edge of $C$ and
%from the silhouette (the locus of points with $z$-vertical tangency) of each $2$-face of $C$, and extend them
%until they hit $\bd C$ (or else all the way to infinity). The resulting subcells have a unique pair of
%faces as their ``floor'' and ``ceiling,'' but their complexity can still be large. 
%In the second decomposition phase, we project each subcell onto the $xy$-plane, apply planar vertical 
%decomposition to the projection, and lift each resulting subcell (trapezoid) vertically up to $\reals^3$ 
%to the range  between the floor and ceiling of the original subcell. This results in a decomposition of 
%$C$ into vertical \emph{pseudo-prisms} (prisms for short), each bounded by up to six facets.
%This recursive scheme (on the dimension) can be generalized to higher dimensions, but it
%becomes more involved as the dimension grows. In this work, though, we only use the three- and four-dimensional 
%scenarios. See~\cite{CEGS-91,Koltun-04a,SA}.

Vertical decompositions, similar to some other geometric decomposition schemes, provide a mechanism for constructing
geometric \emph{cuttings} of various substructures of arrangements of semi-algebraic sets~\cite{AS00}, 
which in turn leads to an efficient divide-and-conquer mechanism for solving a variety of combinatorial and
algorithmic problems, as well as for constructing data structures for geometric searching problems~\cite{AM94}.
The performance of these algorithms and data structures depends on the
complexity (number of prisms) of the vertical decomposition. For $d=2$, the size of the vertical 
decomposition of a cell $C$ is proportional to the combinatorial complexity of (the undecomposed) $C$, 
but already for $d=3$, the size of the vertical decomposition of $C$ can be $\Omega(n^2)$, even when 
the complexity of $C$ is only $O(n)$. In contrast, the complexity of the vertical decomposition of the entire
arrangement in three and four dimensions is near-cubic and near-quartic, respectively, close to the
worst-case bound for the complexity of the undecomposed arrangement (see below for more details).
A challenging problem is thus to obtain sharp bounds on the complexity of the vertical decomposition 
of (the cells of) various substructures of $\A(\S)$ for $d\ge 3$. Despite extensive work on this 
problem, see, e.g., \cite{AAS-97,AES-99,AS-96,dBGH,CEGS-91,Koltun-04a,SS-97} for a sample of results, 
several basic problems remain open.

In this paper we settle some of these open problems in the affirmative, obtaining sharp bounds on 
the complexity of the vertical decomposition of various substructures of arrangements, and full arrangements,
for $d=3,4$. We also provide efficient algorithms for computing vertical decompositions and geometric 
cuttings for such substructures, as well as efficient solutions for the \emph{point-enclosure}
reporting problem in three- and four-dimensional space; see Section~\ref{subsec:pt-encl}
for a formal definition and the details.
%As a major application of these results, we study a variety of proximity problems involving lines/triangles and points in $\reals^3$.  See below for the list of our main results.

\paragraph{Related work.}
Collins~\cite{Col} (see also~\cite{BPR,SS83}) had proposed a mechanism, known as
\emph{cylindrical algebraic decomposition} (CAD), as a general technique for decomposing 
the cells of $\A(\S)$ into pseudo-prisms, for a collection $\S$ of $n$ regions 
in any dimension $d$. However, the number of cells 
produced by a CAD is $n^{2^{O(d)}}$. The vertical decomposition can be viewed as an optimized 
version of CAD, with much smaller complexity. In the plane, the complexity of the vertical 
decomposition of a single cell (resp., of the full arrangement) is proportional to the complexity
of the undecomposed cell (resp., arrangement), so in particular the latter complexity 
(for the full arrangement) is $O(n^2)$.
Vertical decompositions for $d=2,3$ have been used since the 1980's~\cite{CI84,CEGSW90}. A major
progress has been obtained by Chazelle~\etal~\cite{CEGS-91}, who described the construction of 
vertical decompositions in general, for arrangements of semi-algebraic sets in $\reals^d$, and 
proved a bound of $O^*(n^{2d-3})$ on their complexity, for $d\ge 3$ (where the $O^*(\cdot)$ 
notation hides subpolynomial factors). They also showed that the vertical decomposition of 
$\A(\S)$ can be computed in $O^*(n^{2d-3})$ expected time.
The bound was improved to $O^*(n^{2d-4})$, for $d \ge 4$, by Koltun~\cite{Koltun-04a}.
These bounds are nearly optimal for $d\le 4$, but are strongly suspected to be far from 
optimal for $d\ge 5$. Improving the bound, for $d\ge 5$, is a major 20-years-old 
%\micha{20-years-old? (after Vladlen)}
open problem in this area (which we do not address in this work).

In many applications (e.g., motion planning~\cite{SA} and range searching~\cite{Ag:rs,Ma:rph}), one is 
interested in computing the vertical decomposition of (the cells of) only a substructure of $\A(\S)$.
In this case, the goal is to show that if the substructure under
consideration has asymptotic complexity $o(n^d)$, then so should be the complexity of
its vertical decomposition. This statement is true in the plane, as already mentioned.
% and has been shown to hold for arrangements of triangles in $\reals^3$~\cite{dBGH,Tag}. 
There have been a few results 
on bounding the size of the vertical decomposition of substructures of an arrangement
of $n$ semi-algebraic sets in $\reals^3$:
Agarwal~\etal~\cite{ASS-96} showed that the vertical decomposition of the region lying 
between the lower envelope of one family and the upper envelope of another family 
of monotone surfaces in $\reals^3$ (a so-called \emph{sandwich region})
has $O^*(n^2)$ complexity, and Agarwal \etal~\cite{AES-99} showed that for 
any $k < n$, the vertical decomposition of cells of level at most $k$ (namely, cells that have
at most $k$ surfaces passing below them), again for monotone surfaces,
has complexity $O^*(n^2 (k+1))$; see also~\cite{AAS-97,AS-96}. 
De Berg~\etal~\cite{dBGH} showed that the complexity of the vertical decomposition of the arrangement 
of a set of $n$ triangles in $\reals^3$ is $O^*(n^2) + O(\chi)$, where $\chi$ is the number of vertices 
in the arrangement. However, no such output-sensitive bound is known for the arrangement of general 
surfaces, monotone or not.
The most interesting result in this direction, and the most directly related to this paper, 
is by Schwarzkopf and Sharir~\cite{SS-97} who showed that the complexity of the vertical 
decomposition of a single cell in an arrangement of arbitrary semi-algebraic sets of constant
complexity in $\reals^3$ is $O^*(n^2)$, which is near-optimal because the complexity of a 
single cell can be $\Omega(n^2)$~\cite{SA}.
%(which is also an upper bound on the complexity of an undecomposed cell (see~\cite{HS:cell}). 
Notwithstanding these results, the aforementioned fundamental 
problem, of decomposing various substructures in arrangements, has remained largely open for $d\ge 3$.
For example, even though the complexity of the union of a set of objects in $\reals^3$,
in many interesting cases, such as a set of cylinders or a set of (suitably defined) fat objects, is known to be 
$O^*(n^2)$~\cite{AS-00, AEKS-06,Ezra-11,ES-09}, no subcubic bound was known on the size of the vertical 
decomposition of such a union (or of its complement). In $\reals^4$, the complexity of the lower 
envelope of $n$ trivariate functions (whose graphs are semi-algebraic sets of constant complexity) 
is $O^*(n^3)$ (see, e.g.,~\cite{SA}), but no $o(n^4)$ bound was known on the complexity of the 
corresponding vertical decomposition, except in some rather special cases~\cite{AS-96}.

We note that special-purpose decomposition schemes have been proposed for decomposing cells in 
arrangements of hyperplanes, boxes, or simplices, using triangulations, binary space partitions, 
or variants of vertical decomposition; see, e.g., \cite{AS00,ASS21,AS-90,HS96} and references therein. 
Some of these methods also work for arrangements of semi-algebraic sets using the so-called 
linearization technique~\cite{AM94}, albeit yielding in general much weaker bounds.

Due to lack of results on the complexity of vertical decompositions, efficient algorithms for 
computing substructures of arrangements in three and higher dimensions have also remained largely 
elusive. For instance, although the union of a set of balls in $\reals^3$ can be computed in 
near-quadratic time, using the so-called lifting transform, no algorithm with similar performance
is known for computing the union of a set of cylinders in $\reals^3$. Agarwal~\etal~\cite{AAS-97} 
described a randomized algorithm for constructing the vertices, edges, and $2$-faces of
the \emph{minimization diagram} (the $xyz$-projection of the lower envelope)
of a set of trivariate (constant-complexity semi-algebraic) functions 
in $O^*(n^3)$ expected time, and with $O^*(n^3)$ preprocessing it can also compute, in $O(\log n)$ 
time per query, the function that appears on the lower envelope at a query point $\xi\in\reals^3$. 
But their algorithm cannot compute three-dimensional cells of the minimization diagram, nor does 
it compute the vertical decomposition of the minimization diagram. 
Similarly, their algorithm can compute, in $O^*(n^2+\kappa(n))$ expected time,
the vertices, edges, and $2$-faces of the union of a set $\S$ of semi-algebraic sets in $\reals^3$, 
where $\kappa(m)$, as above, is the maximum complexity of the union of a subset of $\S$ of size $m$. 
But their algorithm neither computes the cells of the complement of the union nor the 
  respective vertical decomposition.
Agarwal and Sharir~\cite{AS-96} described an algorithm for computing the minimization diagram of a 
set of trivariate functions assuming that this diagram satisfies certain monotonicity assumptions.

Vertical decomposition has been extensively used, and is often the only mechanism, for computing
geometric cuttings of arrangements, or substructures of arrangements, of semi-algebraic sets in 
$\reals^d$, for $d\ge 2$ (see \subsecref{cut} below for the definition of geometric cuttings), 
which in turn is the fundamental building block of many geometric divide-and-conquer algorithms 
and data structures~\cite{AAEKS,AMS,AAEZ}, and is also used to tackle many combinatorial
problems that exploit such a Divide-and-Conquer mechanism. 
Recently, the development of the \emph{polynomial partitioning} technique~\cite{Guth,GK15}
has lead to an alternative decomposition technique, for the general setup of semi-algebraic regions
of constant complexity in $\reals^d$, for any $d$, and has been successfully applied to 
design geometric divide-and-conquer algorithms and geometric-searching data 
structures~\cite{AAEKS,MP,AAEZ}, on top of many combinatorial results that use this
technique, which have been handled before it has become (effectively) algorithmic in
structures~\cite{AAEKS,MP,AAEZ}. As a matter of fact, for a set of $n$  regions
in $\reals^d$, the polynomial-partitioning scheme produces a decomposition of $O^*(n^d)$ 
cells,\footnote{%
  These cells have constant complexity when the degree of the partitioning polynomial
  is constant, but their shape can have a much more complicated topology.}
a bound that improves upon the best known bound for vertical decomposition for 
$d\ge 5$.  However, a major drawback of polynomial partitioning is that it does not yield improved
bounds for substructures in arrangements, the sort of problems considered in this paper.

%Considering the applications to nearest-neighbor searching with points and lines, as given later in the paper, there is some work on answering nearest-neighbor queries, with query points, amid lines or triangles in $\reals^3$ and higher dimensions. For instance, the nearest neighbor of a point amid $n$ lines in $\reals^3$ can be computed in $O(\log n)$ using a data structure of $O^*(n^3)$ size~\cite{MS}, or in $O^*(n^{3/4})$ time using a linear-size data structure~\cite{AM94}. Agarwal~\etal~\cite{ARS} presented data structures for answer approximate nearest-neighbor queries with points amid simplices in $\reals^d$: for a storage parameter $s$ and an error parameter $\eps>0$, they construct a data structure of size $O^*(s/\eps^d)$, so that an $\eps$-approximate nearest neighbor of a point amid $n$ $k$-dimensional simplices can be computed in $O^*(n/s^{\frac{1}{k+1}})$ time. See also~\cite{Mah15,AM21,AIK09}. In general, semi-algebraic range-searching data structures~\cite{AMS,MP,AAEZ} (see also~\cite{AAEKS}) along with the parametric-search technique (see e.g.~\cite{AM,AM94}) can be used to answer NN-searching queries amid geometric objects by mapping each input object to a point in some high-dimensional parametric space and performing a few range-emptiness queries with appropriate semi-algebraic ranges

\paragraph{Our contributions.}
The paper contains two sets of main results: (i) sharp bounds on the complexity of 
vertical decompositions of several substructures of arrangements in $\reals^3$ and $\reals^4$, and
(ii) efficient algorithms for constructing these decompositions and related structures.
%and (iii) as a major application domain, data structures for various line-point proximity problems in $\reals^3$.

\smallskip
\noindent\textbf{\textit{Vertical decomposition.}} We make significant progress on bounding the 
size of the vertical decomposition of several substructures of arrangements in $\reals^3$ and $\reals^4$. 
%In particular, we first show that the technique by Schwarzkopf and Sharir~\cite{SS-97} can be extended to bound the size of the vertical decompositions of substructures of arrangements in a fairly general setting in $\reals^3$.  Next, we extend our three-dimensional results to some settings in $\reals^4$.

\smallskip
\textit{Union of semi-algebraic sets} (Section~\ref{sec:union}).
%\textit{Monotone substructures of arrangements in 3D.} 
Let $\S$ be a family of $n$ semi-algebraic sets of constant complexity in $\reals^3$. Let $\U(\S)$ denote their union
and $\C(\S)$ the (closure) of the complement of their union. We also refer to $\C(\S)$ as the \emph{free space} of
$\S$, borrowing a notation from algorithmic motion planning. 
For a parameter $m \le n$, let $\psi(m)$ denote the maximum combinatorial complexity of the (undecomposed)
union of a subset of $\S$ of size at most $m$.
(as already mentioned, in many natural applications, $\psi(m)$ is $O^*(n^2)$).
By extending the technique of Schwarzkopf and Sharir~\cite{SS-97}, in a rather non-trivial manner,
we prove that the size of the vertical decompositions of $\C(\S)$ is $O^*(n^2 + \psi(n))$.
In particular, this result implies an $O^*(n^2)$ bound on the complexity of the 
vertical decomposition of $\C(\S)$ whenever the complexity of $\C(\S)$ is $O^*(n^2)$
(e.g.\ $S$ being a set of $n$ cylinders in $\reals^3$ is $O^*(n^2)$), 
considerably improving the previously known nearly-cubic bound.

\smallskip
\textit{Arrangements in 3D} (\subsecref{output_sensitive}).
Let $\S$ be a collection of $n$ semi-algebraic sets of constant complexity in $\reals^3$, and let $\chi$ denote 
the number of vertices in $\A(\S)$. By refining the analysis in the preceding portion of Section~\ref{sec:union}, 
we show that the complexity of the vertical decomposition of the entire arrangement $\A(\S)$ is $O^*(n^2 + \chi)$.

\smallskip
\textit{Lower/upper envelopes in 4D} (Section~\ref{subsec:lower_env}).
Let $\F$ be a collection of $n$ trivariate functions whose graphs are semi-algebraic sets of constant 
complexity, and let $\A(\F)$ denote the arrangement (in $\reals^4$) of their graphs. 
The \emph{lower envelope} $\EE_\F$ of $\F$ is defined as $\EE_\F(x) = \min_{F\in\F} F(x)$, 
for $x\in\reals^3$. The \emph{minimization diagram} of $\F$, denoted by $\M_\F$, is the projection of 
the graph of $\EE_\F$ onto the hyperplane $x_4=0$, i.e., a subdivision of $\reals^3$ (namely, 
of that hyperplane) into maximal connected cells so that the same function appears on the lower 
envelope for all points in a cell. In a fully symmetric manner,
we can define the upper envelope and the maximization diagram of $\F$; see~\cite{SA} 
for more details. We show that the complexity of the vertical decomposition $\VeD{\M}_\F$ 
of $\M_\F$ is $O^*(n^3)$. By extending each three-dimensional pseudo-prism $\tau$ of some
cell of $\VeD{\M}_\F$ to a semi-infinite prism $\hat\tau \subset \reals^4$, where
\[ 
\hat\tau = \{ (\xi,z) \mid \xi\in \tau, z \in (-\infty,\EE_F(\xi)] \}, 
\]
we obtain the vertical decomposition, of size $O^*(n^3)$, of the cell of $\A(\F)$ lying below 
(the graph of) $E_\F$. The same bounds hold for the complexity of the maximization diagram of 
$\F$ and of the cell of $\A(\F)$ lying above the upper envelope.

\smallskip
\textit{Arrangements in 4D} (\subsecref{4darr}).
Let $\S$ be a collection of $n$ semi-algebraic sets of constant complexity in $\reals^4$.
By extending the analysis in Section~\ref{subsec:lower_env},
we show that the complexity of the vertical decomposition of the entire arrangement $\A(\S)$ 
is $O^*(n^4)$. A similar result was obtained by Koltun~\cite{Koltun-04a}, as already mentioned,
but our proof is significantly simpler and an easy extension of the 3D result.

\smallskip
\noindent\textbf{\textit{Algorithms}} (\secref{pt-encl}).
We also present a few algorithmic results, some of which are 
immediate consequences of our combinatorial results.
As is common in computational geometry, we use the \emph{real} RAM model of computation, where we can compute exactly with arbitrary algebraic numbers and each arithmetic operation is executed in constant time.
  Whenever big-O notation and the $O^*(\cdot)$-notation are used, the implicit constant may depend on the description complexity of the semi-algebraic sets in the input.

\smallskip

\textit{Computing vertical decompositions} (\subsecref{algo}).
We present algorithms for constructing vertical decompositions of the kinds mentioned above.
That is, the algorithms construct the set of pseudo-prisms in the vertical decomposition and 
their adjacency graph (connecting pairs of prisms with a common boundary), using the lazy 
randomized incremental approach~\cite{BDS} in time comparable with their respective complexity 
bounds. \subsecref{algo} describes the construction for the complement of the union of semi-algebraic sets 
in $\reals^3$, as well as for the lower envelopes of trivariate functions (whose graphs are semi-algebraic sets of
constant complexity), as an illustration; the same approach extends to sparse 3D arrangements
(see \subsecref{output_sensitive}).

\textit{Geometric cuttings} (\subsecref{cut}).
Let $\S$ be a collection of $n$ semi-algebraic sets of constant complexity in $\reals^d$.
Let $\Pi$ be a substructure of $\A(\S)$, defined by a collection of cells of $\A(\S)$ that 
satisfy certain properties (e.g., lying in the complement of the union or lying below the lower envelope).
For a parameter $r>1$, a \emph{$(1/r)$-cutting} of $\Pi$ (with respect to $\S$) is a set 
$\Xi$ of pseudo-prisms with pairwise-disjoint relative interiors that cover $\Pi$, such that the
relative interior of each pseudo-prism $\tau \in \Xi$ is crossed by (intersected by but not 
contained in) at most $n/r$ sets of $\S$. 
%The subset of $\S$ crossed by $\tau$ is called the \emph{conflict list} of $\tau$.
Our combinatorial results lead to the construction of small-size $(1/r)$-cuttings of $\Pi$. Their
size is dictated by our new bounds for the complexity of the vertical decomposition of $\Pi$,
as made more precise in Section~\ref{sec:pt-encl}.

Specifically, for the case of the complement of the union of semi-algebraic sets in $\reals^3$, 
a $(1/r)$-cutting of size $O^*(r^2 + \psi (r))$ can be computed in $O^*(n(r+\psi(r)/r))$ expected
time, where $\psi(m)$, as above, is the maximum complexity of the (undecomposed) union of a set of at most $m$ 
input sets (of the specific family under consideration). For the case of the region below the 
lower envelope of trivariate functions in $\reals^4$, the bound is $O^*(r^3)$ and the 
expected run time is $O^*(nr^2)$. For the case of an entire three-dimensional arrangement of 
complexity $\chi$, we can compute a $(1/r)$-cutting of $\A(\S)$  of size
$O^*(r^2 + r^3\chi/n^3)$ in $O^*(nr+r^2\chi/n^2)$ expected time. 
In all these results, the respective time bounds also include the cost of computing the
\emph{conflict lists} of each prism of the cutting, namely, the list of all regions that
cross (intersect but do not contain) the prism.

\smallskip
\textit{Point-enclosure queries} (\subsecref{pt-encl}).
Let $\S$ be a family of $n$ semi-algebraic sets in $\reals^3$ of constant complexity.
We obtain a data structure of size and preprocessing cost $O^*(n^2 + \psi(n))$ that, for a 
query point $q\in\reals^3$, detects in $O(\log n)$ time whether $q\in\C(\S)$ and, if not,
returns all $k$ sets of $\S$ containing $q$ in $O(\log{n} + k)$ time. Similarly, for a 
given family $\F$ of $n$ semi-algebraic trivariate functions, we can construct a data structure of
size and preprocessing cost $O^*(n^3)$ that, for a query point $q \in \reals^4$, can report, 
in $O(\log{n} + k)$ time, 
all $k$ functions of $\F$ whose graphs lie below $q$.

  In a companion paper (a preliminary version of which can be found in the earlier conference version of the paper~\cite{AES24}), we present data structures and algorithms for answering nearest neighbor queries, where the input sites as well as query objects are points, lines, or triangles in $\reals^3$.
  These data structures crucially rely on the combinatorial and algorithmic results obtained in this paper.

%--------------------------------------------------------
\section{Vertical decomposition of the complement of the union in 3D} 
\label{sec:union}

Let $\S$ be a collection of $n$ semi-algebraic sets of constant complexity in $\reals^3$.
For the sake of simplicity, we assume that the regions of $\S$ are in general position, in the sense described
in~\cite[Section~7.1]{SA}. This assumption implies, among similar properties, that no two regions or their
$xy$-projections are tangent to each other, that the boundaries of any three of them intersect in $O(1)$ points, and
that the boundaries of no four intersect at a common point. Our results hold even without the 
general-position assumption, as in~\cite{SA}, but the presentation becomes much simpler under this assumption.
By subdividing each set into subsets, if needed,
we may assume that each set $\sigma$ of $\S$ is $xy$-monotone (i.e., any line parallel to the $z$-axis 
intersects $\sigma$ in a connected, possibly empty, interval), that each two-dimensional face of $\sigma$ 
is $xy$-monotone (i.e., any line parallel to the $z$-axis intersects a face at most once) and remains 
openly disjoint from the silhouette and the locus of any $1$-dimensional singularity, and
that each edge of $\sigma$ is $x$-monotone and does not contain any $0$-dimensional singularity.
To subdivide a region $\sigma$ that does not possess these properties, we construct the vertical
decomposition just of $\sigma$ itself. This also subdivides the boundary of $\sigma$ into faces
and edges that have these properties. Note that if an original region $\sigma$ is subdivided
in this manner, its pieces will have in general $z$-vertical faces and edges, and $y$-parallel
edges which are not monotone as above. In what follows we ignore such faces and edges, as their
contribution to the complexity of the vertical decomposition is easy to analyze.
Furthermore, vertical faces of two subregions may overlap, violating the general-position assumption, but we
will be ignoring the vertical faces and the overlap of the boundary of two subregions along their 
vertical faces---this will not cause any difficulty, nor increase the asymptotic bounds in our analysis.

Next, for each region $\sigma$, we partition $\bd\sigma$ into two
parts by splitting it at the points of vertical tangencies (comprising the \emph{silhouette} 
of $\sigma$),
%and the points of singularity, 
and we refer to these parts as the \emph{top} and 
\emph{bottom} (boundary) surfaces of $\sigma$, and denote them as $\bd\sigma^+$ and $\bd\sigma^-$, respectively. 
%\esther{Check footnote.} \footnote{We note that $\bd\sigma$ may have more than two portions if we split at points of singularity, but the number of these portions is at most some constant that depends on the description complexity of $\sigma$. This, however, will not violate our analysis, and we therefore assume without loss of generality that there are only two portions.}

Recall from Section~\ref{sec:intro} that 
the \emph{arrangement} of $\S$, denoted by $\A(\S)$, is the decomposition of $\reals^d$ into 
maximal connected relatively open cells of all dimensions, so that all points within a cell 
lie in the interiors of the same subset of regions of $\S$, and in the relative interiors of the same collection
of faces of region boundaries.
Let $\mu := \mu(\S)$ be the number of vertices in $\A(\S)$. It is well known that the total 
complexity of $\A(\S)$, i.e., the total number of faces of all dimensions in $\A(\S)$, is 
$O(n^2+\mu)$~\cite{SA}. (Informally, any other feature of $\A(\S)$ can either be charged to
a vertex or is defined by at most two regions of $\S$.)

In this section we bound the complexity of the vertical decomposition of the complement $\C(\S)$ of the union 
of $\S$, which consists of some specific subset of cells of $\A(\S)$. Note that $\C(\S)$ is a \emph{monotone} structure, 
in the sense that, for any pair of subsets $\R'\subseteq\R\subseteq\S$, we have 
$\C(\R')\supseteq \C(\R)$.

For a subset $\R\subseteq\S$, let $\psi(\R)$ denote the total complexity of the cells 
in $\C(\R)$, i.e., the overall number of faces, of all dimensions, that bound these cells, and let 
\[
\psi(m) = \max_{\R \subseteq \S,\; |\R|\le m} \psi(\R).
\]
For $k \ge 0$, we say that a point $x$ has \emph{depth} $k$ with respect to a subset $\R$ of $\S$ if $x$ is contained in
the interiors of exactly $k$ regions of $\R$. By definition, all points in a cell of $\A(\R)$ have the same depth, 
and the cells of (the closure of) $\C(\S)$ are exactly those that have depth $0$. 
Let $\subklevel{k}(\S)$ denote the collection of cells of $\A(\S)$ of depth at most 
$k$, let $\psi_k(\R)$ denote the total complexity of the cells of $\subklevel{k}(\R)$, for a subset $\R\subseteq\S$, 
and put $\psi_k(m) = \max_{|\R|\le m} \psi_k(\R)$; by definition, $\psi_0(\R) = \psi(\R)$.
Following the probabilistic argument of Clarkson and Shor~\cite{CS89}, it can be shown that 
\begin{equation}
	\label{eq:CS}
	\psi_k(\S) = O(k^3 \EE[\psi(\R)]),
\end{equation}
where $\R$ is a random subset of $\S$, obtained by choosing 
each set of $\S$ independently with probability $1/k$, and $\EE[\cdot]$ is the expectation over the random choice\footnote{%
  The exponent $3$ comes from the fact that vertices are defined by three surfaces. The number
  of edges and 2-faces is bounded either by charging them to incident vertices, or, if there
  are no such vertices, by a similar probabilistic argument, using a smaller exponent, obtaining
  a bound that is subsumed by the above bound.}
of $\R$. The analysis in \cite{CS89} implies that $\psi_k(n)=O(k^3\psi(n/k))$.

Our goal is to obtain a sharp bound on the size of the vertical decomposition of $\C(\S)$ 
as a function of $n$ and $\psi(n)$. Our main result is the following theorem.

%------------------------------------------
\begin{theorem}
  \label{thm:main}
  Let $\S$ be a collection of $n$ constant-complexity semi-algebraic regions in $\reals^3$, and let $\psi(m)$ be 
  the maximum overall complexity of the cells in $\C(\R)$, over the
	subsets $\R\subseteq \S$ of size at most $m$. Then the total size of the vertical decomposition 
	of the cells in $\C(\S)$ is $O^*(n^2+\psi(n))$.
\end{theorem}
%------------------------------------------

%\paragraph{Remark.}
Except for a few very favorable cases, we typically have $\psi(n) = \Omega(n^2)$, and then
the term $O^*(n^2)$ in the above bound is redundant. The favorable cases include halfspaces, axis-aligned
unit cubes, and, in general, sets $\S$ of regions for which the complexity of $\C(\S)$ is subquadratic (see~\subsecref{output_sensitive}).
It is known that $\psi(m) = O^*(m^2)$ in many cases,
e.g., when $\S$ is a set of spheres, cylinders, arbitrarily oriented cubes, 
homothets of a convex polyhedron of constant complexity, 
or fat convex objects (with a suitable definition of fatness)~\cite{APS}. 
By Theorem~\ref{thm:main}, the size of the vertical decomposition of $\C(\S)$ is $O^*(n^2)$ in these cases.

The proof of Theorem~\ref{thm:main} proceeds through Subsections~\ref{subsec:VD}--\ref{subsec:SS}.
We begin by describing the construction of the vertical decomposition $\VeD{C}$ of a single cell $C$ in 
$\A(\S)$ (Section~\ref{subsec:VD}). We then introduce the notions of vertical visibilities and 
vertical edge crossings (Section~\ref{subsec:visib}), which lie at the heart of the
analysis of bounding the size of the 
vertical decomposition. Finally, we refine and extend the technique of Schwarzkopf and Sharir~\cite{SS-97},
the main technique used in our analysis (Section~\ref{subsec:SS}), to bound the size of the 
vertical decomposition of $\C(\S)$.
We note that this extension is far from being trivial. First, as in our setting, $\S$ consists of full objects rather than monotone surfaces, as considered in~\cite{SS-97}. This leads to new technical difficulties that were not considered in~\cite{SS-97}, and, in particular, implies that one cannot apply the analysis in~\cite{SS-97} as a black box. Moreover, we present two non-trivial extensions to our analysis, the first is with regard to 
%in Section~\ref{subsec:output_sensitive} we extend our analysis to obtain an
output-sensitive vertical decomposition of the entire arrangement (Section~\ref{subsec:output_sensitive}), and in the second we obtain nearly-tight bounds on the complexity of the vertical decomposition of the minimization diagram of trivariate functions (Section~\ref{sec:env}), which is one of the main results in this paper.

%--------------------------------------------
\subsection{Vertical decomposition of a single cell}
\label{subsec:VD}

The vertical decomposition $\VeD{C}$ of a cell $C$ of $\A(\S)$ is constructed in two stages, as follows.
In the first stage, for each edge $e$ of $C$, we erect within $C$ a vertical \emph{wall} (parallel to 
the $z$-axis) from $e$. This wall consists of maximal vertical 
segments contained in (the closure of) $C$ and passing through the points of $e$.
Except for silhouette edges and loci of singularity, 
these segments touch $e$ at one of their endpoints, which is either the top endpoint,
for all points on $e$, or the bottom endpoint. Note that parts of a wall 
may extend to infinity when $C$ is unbounded. These walls partition $C$ into subcells, each of which, denoted $\Delta$,
has a unique pair of faces as its \emph{floor} and \emph{ceiling}---one or both of these faces may be 
undefined when the subcell is unbounded---and all other faces of the subcell lie on the vertical walls that bound it. 
The $xy$-projections of the floor and ceiling faces are identical, and we denote these identical 
$xy$-projections by $\Delta^\downarrow$.
Note that the complexity of $\Delta^\downarrow$ can be arbitrarily large. The subcell $\Delta$ itself is of the form 
\[ 
\Delta = \{ (x,y, z) \mid (x,y) \in \Delta^\downarrow,\, z\in [f^-(x,y),f^+(x,y)] \} ,
\]
where $f^-, f^+$ are bivariate, generally partially defined, semi-algebraic functions of constant 
complexity, whose graphs are portions of the (top or bottom)
boundaries of two of the sets of $\S$, and contain the floor and ceiling faces of $\Delta$, respectively.

The second decomposition stage applies a two-dimensional vertical decomposition to $\Delta^\downarrow$, 
i.e., it erects a $y$-parallel segment from each vertex of $\Delta^\downarrow$ (within $\Delta^\downarrow$), including
singular and locally $x$-extremal points of $\bd\Delta^\downarrow$. 
These segments partition $\Delta^\downarrow$ into pseudo-trapezoidal subcells, called trapezoids for short.
The construction then lifts each resulting trapezoid $\tau^\downarrow$ vertically up in $\reals^3$ in the $z$-direction,
to form the so-called \emph{pseudo-prism} 
\[
\tau = \{ (x,y, z) \mid (x,y) \in \tau^\downarrow,\; z\in [f^-(x,y),f^+(x,y)] \}.
\]
Repeating this step for all subcells created in the first stage, we obtain
a decomposition of $C$ into pairwise disjoint vertical pseudo-prisms 
(prisms for short), each being a semi-algebraic set of constant complexity and 
bounded by up to six facets.  See~\cite{CEGS-91,SA} for further details, and see Figure~\ref{fig:prism}
for an illustration.
%\micha{Add a fig: both 2D layout and 3D prism.}
%\esther{Added, but it needs a fixing, Pankaj made a fig a while ago and he is looking for it.}

%\esther{I fixed Figure~\ref{fig:prism}, pls check.}
%\micha{Fig is nice, but (1) the names (a), (b), (c), (d) are not shown in the figure, and (2) the prism in the last subfig does not seem to bear any resemblance to anything in the other subfigs(?!)}
%\esther{I fixed the labeling. Last prism is supposed to come from a prims in (c), but I don't know how to draw it any better.}

\begin{figure}[htbp]
  \begin{center}
    % \begin{tabular}{cc}
    \includegraphics[scale=0.4]{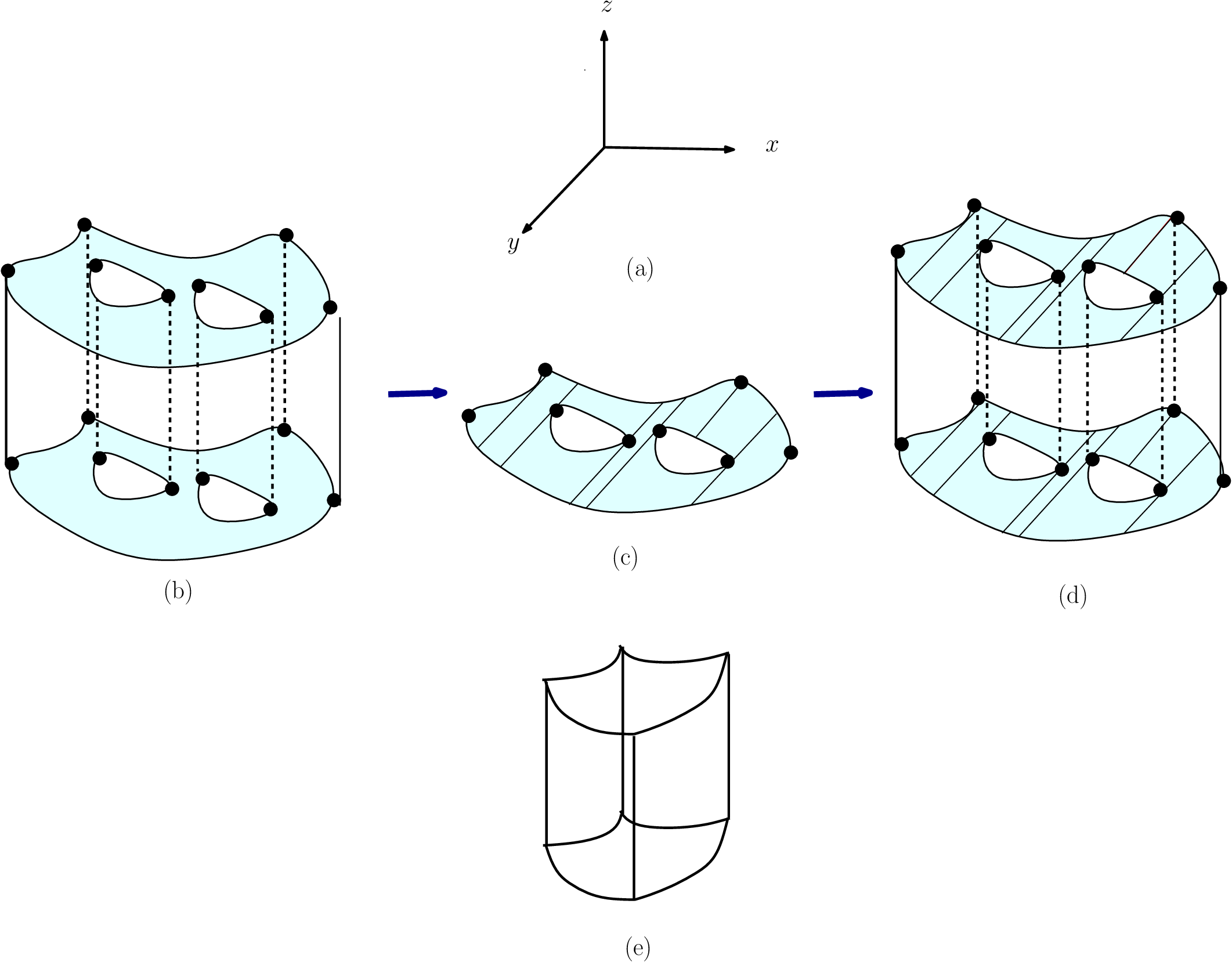}  %&\hspace*{1in}%\\
    %    \end{tabular}
    \caption{(a) The coordinate system of the underlying vertical decomposition.
      (b) A subcell of the first decomposition stage.
      (c) Its $xy$-projection and its 2D vertical decomposition.
      (d) The resulting vertical decomposition of the subcell.
      (e) An example of a final cell. %pseudo-prism of the decomposition.
    }
  \label{fig:prism}
  \end{center}
\end{figure}

The vertical decomposition of a set $\K$ of cells of $\A(\S)$, denoted by $\VeD{\K}$, is the union of
the vertical decompositions of each cell in $\K$. 
Let $|\VeD{C}|$ denote the size of $\VeD{C}$ for a cell $C$, i.e., the number of prisms in $\VeD{C}$, and let 
$|\VeD{\K}| = \sum_{C\in\K} |\VeD{C}|$ denote what we call the size of $\VeD{\K}$.
For any subset $\R \subseteq \S$, let $\VDsubstruc (\R)$ denote the vertical decomposition of 
$\C(\R)$, and let $\VeD{\psi} (\R)$ denote its size. Set 
$\VeD{\psi}(m) = \max_{\R \subseteq \S,\; |\R|\le m}\VeD{\psi}(R)$.

%---------------------------------------------------
\subsection{Vertical visibilities and crossings}
\label{subsec:visib}

For a cell $C$, a triple $\chi = (e, e', g)$ is called a \emph{vertically visible configuration}, or 
\emph{vertical visibility} for brevity, if $e, e'$ are edges of $C$ and $g \subset C$ is a vertical segment 
whose bottom (resp., top) endpoint lies on $e$ (resp., $e'$). We note that, in general, $e$ (resp., $e'$)
lies in the intersection curve $\gamma = \bd\sigma_1 \cap \bd\sigma_2$ 
(resp., $\gamma' = \bd\sigma'_1 \cap \bd\sigma'_2$) of the boundaries of a pair $\sigma_1$, $\sigma_2$
(resp., $\sigma'_1$, $\sigma'_2$) of regions in $\S$, that 
% these two intersection curves and 
the segment $g$ uniquely determines the vertical visibility $\chi$ (the corresponding regions
$\sigma_1$, $\sigma_2$, $\sigma'_1$, $\sigma'_2$ are uniquely determined, assuming general position, as we do),
in the sense that we do not need to know the endpoints of the edges $e$ and $e'$ to define $\chi$. 
This will be important because this vertical visibility will also show up in any subset $\R\subseteq\S$
that includes the four regions $\sigma_1, \sigma_2, \sigma'_1, \sigma'_2$ defining $\chi$, even though
the endpoints of $e$ and $e'$ might vary with $\R$. 
This discussion excludes cases where $e$ and/or $e'$ are silhouette edges or edges of singularity.
Handling visibilities that involve such edges is much simpler, since each such edge is defined
by only one region of $\S$, rather than two. In what follows, we will either ignore such edges
or briefly comment about handling them.

Since we consider cells in the complement of the union, the boundaries $\bd\sigma_1$ and $\bd\sigma_2$ that form $e$ are necessarily top boundaries, i.e., $\bd\sigma_1^+$, $\bd\sigma_2^+$, and the
boundaries $\bd\sigma'_1$ and $\bd\sigma'_2$ that form $e'$ are necessarily bottom boundaries, i.e., $\bd\sigma'^-_1$,
$\bd\sigma'^-_2$ (recall that, after our initial decomposition,  each set $\sigma$ of $\S$ is $xy$-monotone).

Let $\visib(C)$ denote the number of vertically visible configurations in $C$, and let 
$\visib(\S) = \sum_{C\in\C(\S)} \visib(C)$ denote the number of vertically visible 
configurations over all cells $C$ of $\C(\S)$. The same definitions apply, with obvious modifications, 
to any subset $\R\subseteq\S$. It is well known~\cite{CEGS-91}, and easy to show, that, for any cell $C$, 
$|\VeD{C}|$ is proportional to the number of vertical visibilities in $C$ plus the complexity of $C$. 
In particular, replacing $\S$ by a subset $\R$, we get:
%-------------------
\begin{lemma}
	\label{lem:VD-visib}
	For any subset $\R\subseteq \S$ and for any subset $\K$ of cells of $\A(\R)$,
	\[ |\VeD{\K}| = \sum_{C\in\K} O(\visib (C)  + |C|) . \]
\end{lemma}
%------------------
Therefore $\VeD{\psi}(\S) = O(\visib(\S)+\psi(\S))$.
It thus suffices to bound the number $\visib(\S)$ of vertical visibilities in the cells of $\C(\S)$. 

We extend the notion of vertical visibility to
\emph{vertical edge crossing}, or crossing for short, in which the segment $g$ may be crossed by
boundaries of some regions of $\S$, and/or be fully contained in some regions. 
A vertical crossing is a triple $\chi = (e,e',g)$, as above, but $e$ and $e'$ are now two arbitrary edges of $\A(\S)$.
Note though that, in the context of analyzing $\C(\S)$, we will only consider crossings for which $e$ is formed by two top boundaries, and $e'$ by two bottom boundaries.
Our analysis will show that these vertical crossings are not fully contained in a region from $\S$---see below.
Let $\defn(\chi) =\{\sigma_1, \sigma_2, \sigma_3, \sigma_4\} \subseteq \S$, the \emph{defining set} of $\chi$, 
denote the (multi)set of four elements of $\S$ such that $e \subseteq \bd\sigma_1^+ \cap \bd\sigma_2^+$
and $e' \subseteq \bd\sigma_3^- \cap \bd\sigma_4^-$. If $e$ is a portion of an edge of a set of $\S$ 
(namely a silhouette edge or an edge of singularity), we encode it by saying that $\sigma_1=\sigma_2$, 
and similarly for $e'$. 
% and a similar encoding can be used for representing loci of singularity.
Like a vertical visibility, $\chi$ is defined by the vertical segment $g$, which uniquely 
determines the two intersection curves $\bd\sigma_1^+\cap\bd\sigma_2^+$ and $\bd\sigma_3^-\cap\bd\sigma_4^-$.
Note also that any pair of such curves determine up to $s$ crossings, for a suitable constant parameter $s$ that
depends on the complexity of the regions of $\S$; again, the endpoints of $e$ and $e'$ do not play a role. 

We define the \emph{level}\footnote{%
	We \emph{warn} the reader that our analysis will also involve other notions of levels; see later.}
of a crossing $\chi$ with respect to $\R\subseteq \S$ to be the number of 
regions of $\R$ that intersect $g$. Each such intersection is an interval along $g$, each of whose 
endpoints is either an endpoint or an interior point of $g$.

%----------------------------------------------
\subsection{Extending the Schwarzkopf-Sharir analysis}
\label{subsec:SS}

We are now ready to extend the analysis of Schwarzkopf and Sharir~\cite{SS-97} to prove that 
	$\varphi(\S) = O^*(|\S|^2+\psi(\S))$. We review the 
(fairly complicated) technique of \cite{SS-97}, cast in our more general and somewhat different setting, with a focus on
the (numerous) ingredients of their technique that one needs to modify in order to extend the analysis to our setting.
The general strategy in \cite{SS-97} is to fix some sufficiently large near-constant parameter $k$
(choosing $k$ to be a constant will not do; see below for details concerning the choice), and to 
apply a two-stage global \emph{charging scheme} for vertical visibilities, which eventually 
leads to the following recurrence,
where $\visib_j(\R)$ is the number of vertical crossings at level $\le j$ for a subset $\R$ of $\S$, and 
$\visib_j(m) = \max_{\R\subseteq \S,\;|\R|\le m} \visib_j(\R)$:
\[
\visib(n) = \visib_0(n) = O\Bigl(k^3\beta^2(n) (\psi(n) + n^2) + k^2\beta^3(n) \visib(n/k) \Bigr),
\]
where $\beta(n)$ is a very slowly growing near-constant function of $n$ related to the inverse Ackermann function.
It is in fact $O(\lambda_{s+2}(n)/n)$, where $\lambda_{s+2}(n)$ is the maximum length of a Davenport-Schinzel sequence 
of order $s+2$ on $n$ symbols~\cite{SA}, and $s$ is the constant parameter introduced earlier.
Standard analysis, that we detail later, shows that this recurrence solves to $\visib(n) = O^*(\psi(n) + n^2)$,
as asserted in Theorem~\ref{thm:main}.
% In other words, the number of vertical edge crossings of level at most $6k$ in $\C(\S)$ is larger 
% than the number of vertical visibilities by at least a factor close to (somewhat smaller than) $k^2$, 
% up to an overhead term that is the sum of a term close to $n^2$, plus a term that is nearly proportional to the % complexity of $\C(\R)$,
% for a sample $\R$ of size $n/k$. 

Before we describe the charging scheme, we perform a preprocessing 
step that splits some of the edges of $\A(\S)$, so that some desirable properties hold.

%-------------------------------
\paragraph{Preprocessing step.}
First, by dividing each edge of $\A(\S)$ into $O(1)$ pieces, as needed, we can assume that every 
edge of $\A(\S)$ is $x$-monotone (see also the discussion at the beginning of Section~\ref{sec:union}). 
For an edge $e$ of $\A(\S)$, the upward (resp., downward) 
\emph{vertical curtain} $V_e$ (resp., $V^e$) erected from $e$ is the union of 
all vertical rays emanating upwards (resp., downwards), in the $z$-direction, from the points of $e$
($e$ appears as a subscript (resp., superscript) to indicate that it lies at the bottom
(resp., top) of the curtain). 
Since each set $\sigma \in \S$ is $xy$-monotone, $\bd\sigma \cap V_e$ (or $\bd\sigma\cap V^e$) 
consists of \emph{top} and \emph{bottom} boundary curves, each of which might %be empty or
consist of $O(1)$ connected components, so that each $z$-vertical ray in $V_e$ or $V^e$ intersect 
each curve at most once. 
%\esther{Specifically, $\bd\sigma \cap V_e$ consists of both top and bottom boundary curves, since $e$ lies on $\C(\S)$, whereas $\bd\sigma\cap V^e$ may consist of only one boundary curve.}
Clearly, when both curves exist (for the same set $\sigma$), the bottom boundary curve lies below the top boundary curve. 
%\micha{(a) Say what happens when only one boundary crosses? and (b) Add a figure?}
%\esther{(a) This implies that $e$ lies inside the cross section of a set  $\sigma \in \S$ with the vertical curtain  of $e$, no? (b) and this is what the figure should describe.}
%\esther{Concerning vertical curtains of type $V_e$, if there is only one boundary curve that crosses it, then it must be a lower boundary. What do we have to say on vertical curtains of type $V^e$? If there is only one boundary then it should be a top boundary?}
Let $\Sigma_e$ (resp., $\Sigma^e$) 
denote the set of these bottom (resp., top)  boundary curves within $V_e$ (resp., $V^e$), and let $\A(\Sigma_e)$
(resp., $\A(\Sigma^e)$) denote their two-dimensional 
arrangement within $V_e$ (resp., $V^e$). We define the \emph{V-level} of a point $p$ in $V_e$ 
(resp., $V^e$) to be the number of arcs in $\Sigma_e$ (resp.,  $\Sigma^e)$ that lie below (resp., above) $p$.

We split the edges of $\A(\S)$ to ensure that, over each subcurve $e$, none of the curves 
of $\Sigma^e$ or $\Sigma_e$ has an endpoint (namely, $x$-extremal point, silhouette point, or 
singular point) within the first $6k$ levels of the planar arrangements $\A(\Sigma_e)$ (within the 
relative interior of $V_e$) or $\A(\Sigma^e)$ (within the relative interior of $V^e$). 
As argued in~\cite{SS-97}, using an analysis that also holds in the more general setup considered here,
the overall increase in the number of these split edges, over all edges $e$ of $\A(\S)$, is 
$O(kn^2\beta(n/k)) = O(kn^2\beta(n))$. 
Hence, the total number of edges in $\C(\S)$ after this splitting step is $O(\psi(\S)+kn^2\beta(n))$. 

\begin{figure}[htb]
  \begin{center}
  \includegraphics[scale=0.4]{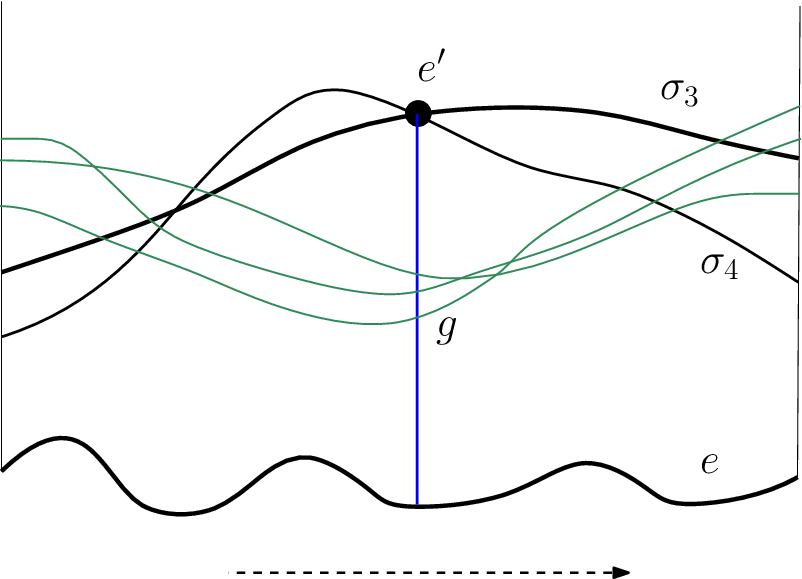}
  %\vspace{1in}
  \caption{%Draw a figure showing the vertical curtain, split edge, $\sigma(\chi)$...
    A vertical crossing $\chi_v=(e,e',g)$ with $\sigma(\chi)=\sigma_3$.}
  \label{fig:split-edge}
  \end{center}
\end{figure}

Fix a (split) edge $e$ of $\C(\S)$ and assume that $e$ is a lower edge, so $\C(\S)$ lies above $e$, locally near $e$,
and $e$ is contained in the intersection of two top region boundaries. Consider the upward 
curtain $V_e$. A vertex $v\in \A(\Sigma_e)$ corresponds to a vertical crossing $\chi_v=(e,e',g)$, where 
$v$ is the top endpoint of $g$ and lies on the edge $e'$ that crosses $V_e$ at $v$, and $e$ 
contains the bottom endpoint of $g$. Suppose that $e' \subseteq \bd\sigma_3^- \cap \bd\sigma_4^-$,
for a pair of regions $\sigma_3$, $\sigma_4\in\S$ (by definition and construction, the incident boundaries 
must be bottom boundaries, so they do appear in $\Sigma_e$), and that the incident 
portion of $\bd\sigma_3^-$ lies above $\bd\sigma_4^-$ to the right of $v$ in a sufficiently small neighborhood 
of $v$ within $V_e$. We then set $\sigma(\chi)=\sigma_3$; otherwise we set $\sigma(\chi_v)=\sigma_4$. 
%\micha{Add a figure?}
Informally, $\sigma(\chi)$ is the region whose boundary is hidden from $e$ by the other boundary, as we trace it to
the right from $v$ within $V_e$.
See Figure~\ref{fig:split-edge}.

%--------------------------------------
\paragraph{First stage.}
In the first charging step, roughly speaking, we fix a (split) lower edge $e$ of a cell of $\C(\S)$, 
and construct a set $\CSet_e$ of vertical edge-crossings of level at most $k$, that use $e$ as their lower edge
(so they all lie within $V_e$), 
and bound the overall number $\visib_e(\S)$ of vertical visibilities (within $\C(\S)$) that use $e$ as their
lower edge, in terms of $|\CSet_e|$.
% (\pankaj{Need to check whether the edges of $\subklevel{\x24}(\S)$ are enough or whether we should increase to $\le 4$.})
Set $\CSet := \bigcup_e \CSet_e$, over all these edges $e$. 
Observe that $\visib(\S) = \sum_e \visib_e(\S)$.

Let $\delta$ be a portion of a lower edge $e\in \C(\S)$, and let $\sigma_1$, $\sigma_2\in\S$
be the two regions such that $e\subseteq \bd\sigma_1^+\cap\bd\sigma_2^+$.
A set $\sigma \in \S$ is \emph{relevant} for $\delta$ if the bottom surface $\bd\sigma^-$ appears on the lower
envelope of $\Sigma_e$ over some point $p$ of $\delta$, i.e., the 
vertical ray from $p$ in the $(+z)$-direction intersects 
$\bd\sigma^-$ before intersecting any other set of $\S$.

We take each (split) edge $e$ and partition it into a right portion $e_R$ and a left portion $e_L$, where 
$e_R$ is the smallest right portion of $e$ such that $k$ distinct sets of $\S$ are relevant for $e_R$. 
(If fewer than $k$ sets are relevant for $e$ then $e_R=e$ and $e_L$ is empty.) 
By a standard lower envelope argument, the number of vertical visibilities 
that emanate from (i.e., their bottom endpoints lie on) the right portion $e_R$ of an edge $e$ is then $O(k\beta(k))$
\cite{SA}. Summing over all (split) edges of $\C(\S)$, the contribution of the right portions to $\visib(\S)$ is 
$$O\Bigl(k\beta(k)(\psi(\S)+kn^2\beta(n))\Bigr).$$ 
It thus suffices to bound the number of vertical visibilities
whose bottom endpoints lie on the left portions $e_L$ of  edges $e$ (for which $e_L$ is non-empty). 
We write the set of these visibilities as $\CSet_{e_L}(\S)$, and denote its size as $\visib_{e_L}(\S)$.
Our goal is to estimate $\sum_e \visib_{e_L}(\S)$.

Fix a pair of sets $\sigma_1, \sigma_2 \in \S$, and let $\gamma$ be one of the $x$-monotone 
connected components of $\bd\sigma_1^+ \cap \bd\sigma_2^+$. Define 
\[ 
E_\gamma = \{ e \mid e\; \mbox{is a lower (split) edge of\, } \C(\S) \;\wedge\; 
e \subseteq \gamma \;\wedge\; e_L \ne \emptyset\} ,
\]
% \micha{Rewrite the whole para}
and put $\nu_\gamma = |E_\gamma|$. Let $\CSet(\gamma)$ be the set of vertical visibilities whose bottom endpoints
lie on the left portion of an edge of $E_\gamma$. Let $V_\gamma$ be the union of the vertical curtains $V_{e_L}$
over all left portions $e_L$ of edges $e$ in $E_\gamma$ (so $V_\gamma$ is contained in the vertical curtain erected
from $\gamma$, and its sub-curtains $V_{e_L}$ are pairwise disjoint and not adjacent to one another). 
Let $\S_\gamma\subseteq \S$ be the family of regions that are relevant for the left portion $e_L$ of some edge
$e\in E_\gamma$, and let $\Sigma_\gamma$ be 
the family of connected components of the intersection curves $\bd\sigma^- \cap V_{e_L}$, 
over all edges $e\in E_\gamma$, and over all $\sigma\in\S_\gamma$. Set $t_\gamma := |\S_\gamma|$. 
Informally, we treat $\gamma$ as a single entity, but focus only on the left portions $e_L$ of the edges $e\in E_\gamma$.

Let $\chi=(e,e',g)$ be a vertical visibility whose bottom endpoint lies on 
the left portion $e_L$ of some edge $e\in E_\gamma$. Then the top endpoint of $g$ is a 
breakpoint $v$ of the lower envelope of $\Sigma_\gamma$. 
Hence, the overall number of such vertical crossings, denoted by $\visib_{\gamma}(\S)$,
is at most $\lambda_{s+2}(t_\gamma)$, where $s$ is the previously defined parameter that
depends on the complexity of the sets in $\S$.
Using our notation, we write this as $O(t_\gamma\beta(t_\gamma))$.

\begin{figure}[htb]
  \begin{center}
	  \begin{tabular}{ccc}
		  \includegraphics[scale=0.4]{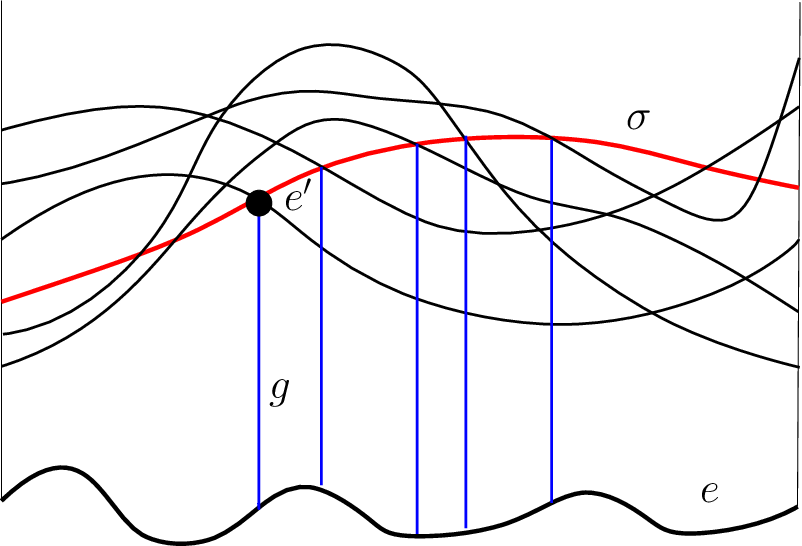} & \hspace*{1in} & \includegraphics[scale=0.4]{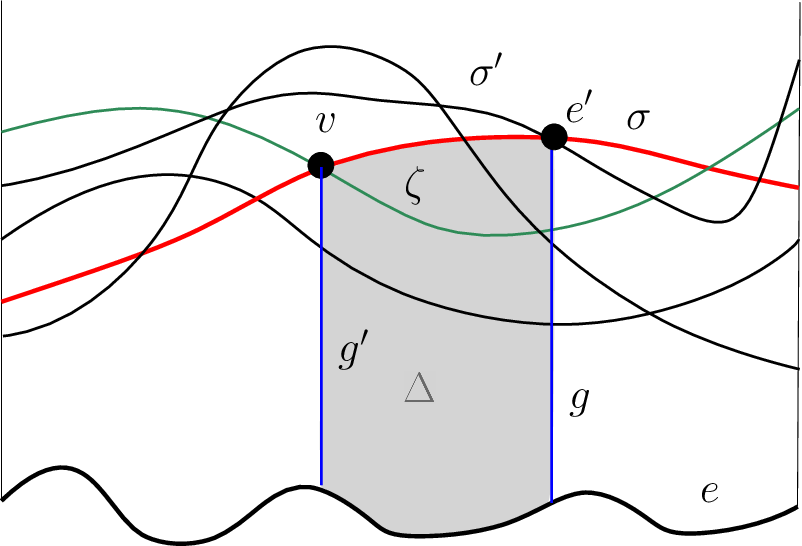}\\ %\hspace*{1in}\\
			(a)& \hspace*{1in} &(b)
	  \end{tabular}
	  \caption{(a) The setup in the construction of $\CSet_e$; (b) illustration of property (R4).}
    \label{fig:Re}
  \end{center}
\end{figure}

Now consider a set $\sigma \in \S_\gamma$, and let $p \in \bd\sigma^-\cap V_\gamma$ be a point 
(not a vertex) on the lower envelope of $\Sigma_\gamma$. Then $p\in \bd\C(\S)$.
Suppose $p$ lies above $e_L$ for some edge $e \in E_\gamma$.
We collect a set of at least $k-1$ vertices of $\A(\Sigma_\gamma)$ of V-level at most $k$, as described below.
We apply this step starting from only
one point $p$ on $\bd\sigma^-\cap V_\gamma$, for each $\sigma\in \S_\gamma$.
We start from $p$ and follow the curve $\bd\sigma^- \cap V_e$ in increasing $x$-direction. 
When we encounter the bottom boundary of a new distinct set $\sigma'$, i.e., reach the first 
intersection point $v\in\bd\sigma^-\cap\sigma'$ to the right of $p$, we pick $v$. (Since $p\in\bd\C(\S)$, the first time
we meet $\sigma'$ is necessarily on its bottom boundary.) We stop as soon as we 
have collected $k$ vertices or reached the right boundary of $V_{e}$, see Figure~\ref{fig:Re}~(a).
(We may also traverse the portion of $\bd\sigma^- \cap V_e$ lying above $e_R$.)
For each vertex $v$ that we collect, let $\chi_v$ be the vertical crossing with the lower edge $e$
and top endpoint $v$.
Note that $\sigma(\chi_v)=\sigma$, and therefore $v$ is only collected when we traverse $\bd\sigma^-\cap V_\gamma$ 
(and not when we traverse $\bd{\sigma'}^- \cap V_\gamma$, where $\sigma'$ is the other surface incident to $v$, as defined above).
The level of $\chi_v$ is at most $k$.
Let $\CSet(\gamma,\sigma)$ denote the set of vertical crossings corresponding to the collected vertices.
Since $k$ surfaces are relevant for $e_R$, i.e., they appear on the lower envelope over $e_R$,
and at least $k-1$ of them (all except $\sigma$ if it is relevant for $e_R$) lie above $\bd\sigma^-$ at 
(the starting point) $p$, 
$|\CSet(\gamma,\sigma)| \ge k-1$.

We repeat this step for all regions $\sigma\in\S_\gamma$. 
Set $\CSet(\gamma)=\bigcup_{\sigma\in\S_\gamma} \CSet(\gamma,\sigma)$ and
$\CSet=\bigcup_\gamma \CSet(\gamma)$, where the union is taken over all $x$-monotone  
connected components $\gamma$ of intersection curves $\bd\sigma^+ \cap \bd\sigma'^+$ (that show up on $\bd\C(\S)$), for $\sigma,\sigma' \in \S$. The argument just given implies the following:
\begin{equation}
  \label{eq:phi-gamma} 
  |\CSet(\gamma)| \ge (k-1) |\S_\gamma| = (k-1)t_\gamma =
  \Omega \left(k \frac{t_\gamma}{\lambda_{s+2}(t_\gamma)}\cdot \lambda_{s+2}(t_\gamma)\right ) 
  \ge \frac{c'k}{\beta(t_\gamma)} \cdot \visib_{\gamma}(\S), %\varphi(\gamma) , 
\end{equation}
for a suitable constant $c'$. Summing over all intersection curves $\gamma$, observing that, by definition, each  
vertical visibility over the left portion of a lower edge is counted in the sum of the right-hand
sides of (\ref{eq:phi-gamma}) exactly once,
and adding the bounds on the number
of vertical crossings over the right portions $e_R$ of the edges $e$ (note that these vertical crossings are counted at most twice in this charging), we obtain the following:
% \micha{I don't understand the equation.}
% \micha{Isn't $\Phi$ here an overloading of notation? Replace by, say, $\Xi$?}
\begin{align}
  \visib(\S)  = \sum_{\gamma} \visib_\gamma(\S) & \le \frac{c\beta(n)}{k} \sum_\gamma |\CSet(\gamma)| +
  O\Bigl(k\beta(k)(\psi(\S)+kn^2\beta(n))\Bigr)  \nonumber \\
  &=
  \frac{c\beta(n)}{k} |\CSet| + O\Bigl(k\beta(k)(\psi(\S)+kn^2\beta(n))\Bigr) ,
  \label{eq:R}
\end{align}
where $c = 1/c'$ is a constant.

As proved in~\cite{SS-97} (in a different context, which also applies in our setup),
besides the inequality~(\ref{eq:R}), $\CSet$ satisfies the following properties:
\begin{itemize}
	\item[(R1)] Each $\chi\in\CSet$ is at level at most $k$.
	\item[(R2)] For each $\chi=(e,e',g)\in\CSet$, the edge $e$ is an edge of $\C(\S)$ 
		such that a cell of $\C(\S)$ lies locally immediately above $e$ (and $e$ is incident to two top boundaries).
	\item[(R3)] For any three sets $\sigma_1, \sigma_2, \sigma_3 \in \S$, there are at most 
		$O(k)$ vertical crossings $\chi=(e,e',g)\in \CSet$ with 
		$e \subseteq \bd\sigma_1^+ \cap \bd \sigma_2^+$ and $\sigma(\chi)=\sigma_3$.
		This immediately follows from the rules by which the sets $\Phi_e$ are constructed,
		namely that we start from only one visibility $\chi$ with a lower edge $e$ contained in (an $x$-monotone 
  connected component of) any given intersection curve $\bd\sigma_1^+\cap \bd\sigma_2^+$ and has 
  $\sigma(\chi) = \sigma_3$.
		% \micha{See an earlier comment, needed to ensure that (R3) holds.}
	\item[(R4)] % \micha{Check that (R4) also includes the case where $\sigma(\chi)$ is incident to $e$.}
	    For each $\chi=(e,e',g)\in \CSet$, and for any set $\zeta\in\S$ whose bottom boundary $\bd\zeta^-$ 
		intersects the relative interior of $g$, there is an intersection point $v$ of % \micha{top? bottom?}
		$\bd\zeta^-$ and $\bd \sigma^- = \bd \sigma(\chi)^-$ on $V_e$ to the left of $\chi$, so that the 
		portion of $\bd\zeta^-\cap V_e$ between $v$ and $g$ does not meet $\bd\sigma^-$, and the portion of 
		$\bd\sigma^- \cap V_e$ between $v$ and $g$ does not meet the boundary $\bd\sigma'^-$ of the other 
		set $\sigma'$ incident on the top endpoint of $g$ (this latter property is an immediate consequence of
		the rule by which we charge crossings).
		In addition, if $g'$ denotes the 
		vertical segment connecting $e$ and $v$, then the trapezoid $\triangle_v$ formed 
		within $V_e$ by $g, g'$, the part of $e$ between $g$ and $g'$, and the part of 
		$\bd\sigma^-\cap V_e$ between $g$ and $g'$, is intersected by the boundaries of at 
		most $k$ sets of $\S$. See Figure~\ref{fig:Re}~(b).
\end{itemize}

Properties (R1)--(R3) are easy to establish ((R3) has already been argued).
The proof of (R4) proceeds as in \cite{SS-97}, with
suitable straightforward modifications.

%-----------------------------
\paragraph{Second stage.}
% \micha{Check $6k$.}
In the second stage, we bound $|\CSet|$ in terms of $\visib_{4k}(\S)$ 
(using $4k$ instead of $k$ is done for technical reasons, and is not important in the final 
analysis). 
%The analysis in this step, as in \cite{SS-97}, is more global and considerably more involved.
%(Strictly speaking, there is no real charging in the second stage, and it is more like an accounting scheme 
%(as it is in \cite{SS-97}).) 
We now switch the roles of edges, and fix an upper (split) edge $e'$ of $\A(\S)$ of depth 
at most $k$
(i.e., separated from $\C(\S)$ by at most $k$ sets; all points on $e'$ have the same depth
after the preprocessing step), incident on two bottom boundaries of regions in $\S$. 
Let $\CSet^{e'}\subseteq \CSet$ be the set of the vertical crossings collected at the first 
stage that have $e'$ as their upper edge.
Note that for any edge-crossing $(e,e',g)\in \CSet^{e'}$, the bottom endpoint of $g$ (on $e$) 
lies on $\bd\C(\S)$. This implies that for any set $\sigma\in\S$ that crosses $g$, its bottom
boundary $\bd\sigma^-$ crosses $g$ (and its top boundary may or may not cross $g$).

We seek an upper bound on $|\CSet^{e'}|$. Some of the crossings in $\CSet^{e'}$ will be counted locally,
for each edge $e'$ separately, while others will be counted in a global manner, considering all edges together.
To carry out the analysis, we classify each vertical crossing $\chi=(e,e',g)\in\CSet^{e'}$ as either covered or
uncovered. We call $\chi$ \emph{covered} if the following condition holds:

\begin{quote}
  Let $\sigma_1, \sigma_2 \in \S$ be the two sets whose top boundaries contain $e$. There is a set $\zeta\in\S$ 
  whose bottom boundary $\bd\zeta^-$ intersects the relative interior of $g$ (the top boundary may 
  also cross $g$ but we ignore it; as already noted, 
  the bottom endpoint of $g$ does not lie in 
  any set of $\S$, so $g$ cannot be crossed by just the top boundary).
  The condition for being covered is that
  (i) one of $\bd\sigma_1^+$, $\bd\sigma_2^+$, say $\bd\sigma_1^+$, crosses $\bd\zeta^-$ within $V^{e'}$,
  to the left or to the right of $g$, at some point $w$, and (ii) if $g^*$ denotes
  the vertical segment connecting $w$ to $e'$, then the pseudo-trapezoid $\triangle_w$ formed 
  by $g$, $g^*$, and the portions of $e'$ and $\bd \sigma_1^+$ between $g$ and $g^*$, is crossed by the
  boundaries of at most $2k$ sets of $\S$.
\end{quote}

(Informally, at most $k$ of these regions cross $g$, and we do not want more than $\approx k$ other regions
to cross the relevant portion of $\bd\sigma_1^+$ between $g$ and $g^*$.)
See Figure~\ref{fig:covered}. 
Otherwise $(e,e',g)$ is \emph{uncovered}. (Informally, in this case it takes `longer' for $\bd\zeta^-$
to reach $\bd\sigma_1^+$, $\bd\sigma_2^+$.) We first bound the number of uncovered vertical crossings 
within $V^{e'}$ using a ``local'' recursive argument, which relates this
number to the overall number of vertices of $\A(\Sigma^e)$ of V-level at most $2k$ within $V^{e'}$. We then 
use a global planarity argument to bound the number of covered crossings.
Both arguments are adapted from \cite{SS-97} and turn out to be
more involved because of the presence of both top and bottom boundary surfaces;
we provide here many of the details, adapted to our context, and refer the reader to~\cite{SS-97} for the missing details.

%\micha{In Fig. 4: Part (b) is wrong: $\bd\sigma_2^*$ does NOT cross. Fix! I also suggest to draw $\bd\sigma_2^*$ in red in part (a) too, and to draw $\bd\sigma_2^*$ in blue in both parts.}
%\esther{I made some modifications, but I am not sure what you intend to illustrate in (b). To be discussed.}

%\esther{Fixed Fig(b). Pls check figs and caption.}

\begin{figure}[htb]
  \begin{center}
    \begin{tabular}{ccc}
      \includegraphics[scale=0.5]{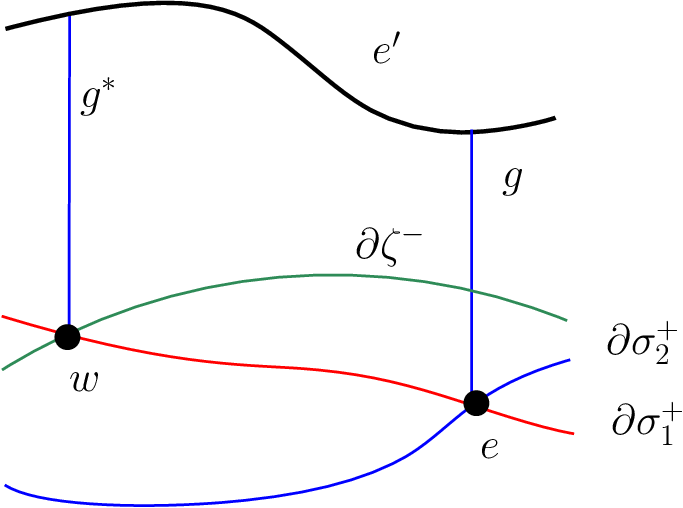} & \hspace*{0.5in} & \includegraphics[scale=0.5]{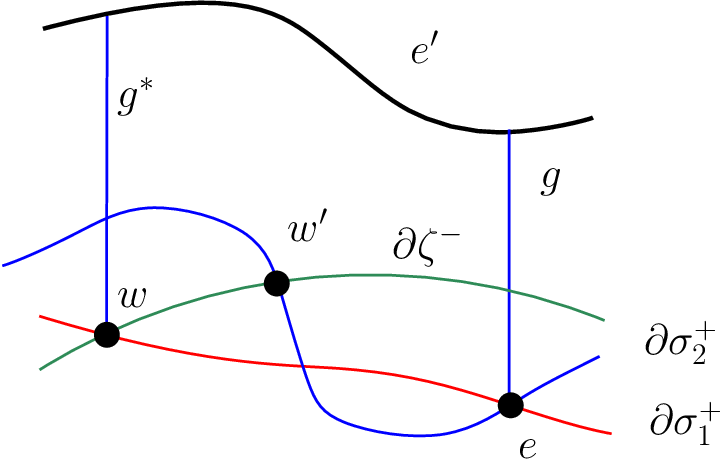}\\ %\hspace*{1.5in}\\
      (a) & \hspace*{1in}& (b)
    \end{tabular}
    \caption{A covered edge crossing $g$. In this figure the relevant $e$-incident surface is $\bd\sigma_1^+$.
      (a) $\bd\sigma_2^+$ crosses $\bd\sigma_1^+$ once before $w$ (at $e$).
      (b) $\bd\sigma_2^+$ crosses $\bd\sigma_1^+$ twice before $w$.
    }
    \label{fig:covered}
  \end{center}
\end{figure}

%----------------------------------------------
%\textbf{\textit{Uncovered vertical crossings.}}
\paragraph{Uncovered vertical crossings.}
The analysis of uncovered edge-crossings is performed separately for each (split) edge $e'$.
Let $e'$ be a split edge of $\A(\S)$ at depth at most $k$ with $\CSet^{e'} \ne \emptyset$. 
Let $t=t_{e'}$ be the number of sets of $\S$ whose top boundaries appear on the first $2k$ V-levels 
(highest $2k$ V-levels from $e'$ down) of $\A(\Sigma^{e'})$.
(Recall that the V-level of a point $p$ in $\Sigma^{e'}$ is the number of top boundary curves of 
$\Sigma^{e'}$ that appear above $p$ (and below $e'$); observe that bottom surfaces are ignored here.) 
The preprocessing step ensures that no endpoint (namely, a
locally $x$-extremal point or  a singular point)
of a curve of $\Sigma^{e'}$ lies in the relative interior of $V^{e'}$ at any V-level $\le 6k$. 

At $g$ itself, the only boundaries that cross it came ``from above,'' as we have reached $g$ while tracing $e$ and
$V_e$ in the first stage. Each of these crossing regions has its bottom boundary cross $g$, and some may also have
their top boundary cross $g$ too. As we trace $e'$ and $V^{e'}$ in the second stage, 
nothing else can reach the moving $g$ from above, as $e'$ is an edge (and nothing that now crosses $g$ can
escape through the upper edge $e'$), 
so anything new that penetrates $g$ comes from below, and contributes a crossing top boundary (and some may also 
contribute, later, their bottom boundary).
But, as just said, $g$ also has at least one bottom boundary that reached it
from above. In the covered situation, one of these ``older'' bottom boundaries, $\bd\zeta^-$, crosses an $e$-incident 
top boundary $\bd\sigma_1^+$ at crossing distance (number of separating surfaces) at most $k$ from $g$. 
This does not happen in the uncovered situation, meaning that as we trace
$\bd\sigma_1^+$ from $g$, the first $k$ vertices that we encounter are formed by new top boundaries that cross
$\bd\sigma_1^+$ from below. Technically, this means that we may ignore all the old top boundaries that cross $g$, 
because they cannot reach $\bd\sigma_1^+$ in the uncovered case (informally, they stay ``shielded'' from $\bd\sigma_1^+$ 
by their corresponding bottom boundaries, which do not reach $\bd\sigma_1^+$). See Figure~\ref{fig:uncovered_edge}.

If $t \le 6k$, $e'$ contributes only $O(k^2)$ vertical crossings to $\CSet$.
Summing over all (split) edges $e'$ of $\subklevel{k}(\S)$, the overall number of vertical crossings in $\CSet$ 
within all the curtains $V^{e'}$ with $t_{e'} \le 6k$ is 
\[
O(k^2(k n^2\beta(n) + \psi_k(\S)) ) = 
O(k^3n^2\beta(n) + k^2\psi_k(\S)) ,
\]
which, by the Clarkson-Shor technique (see~(\ref{eq:CS})), is 
\[
O(k^3n^2\beta(n) + k^5\EE[\psi(\R)]),
\]
where $\R$ is a random subset of $\S$ obtained by choosing each element of $\S$ with probability $1/k$, 
and $\EE$ is the expectation with respect to this random choice.
%\esther{This should be $O(k^3n^2\beta(n) + k^5\EE[\psi(\R)])$, by~(\ref{eq:CS}).}
We may thus assume that $t > 6k$.

\begin{figure}[htb]
  \begin{center}
    \includegraphics[scale=0.5]{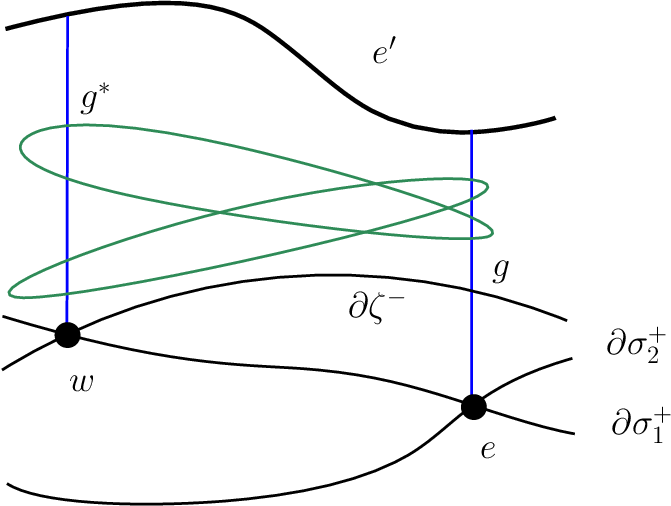}
    %\vspace{1in}
    \caption{The setup at an uncovered edge crossing.}
    \label{fig:uncovered_edge}
  \end{center}
\end{figure}

The arcs of $\Sigma^{e'}$ can be partitioned and clipped into $O(t\beta(t))$ subarcs, so that they cover 
the first $k$ V-levels of $\A(\Sigma^{e'})$, and each subarc contains at most $k$ vertices of 
$\A(\Sigma^{e'})$. To achieve this, we observe that there is a value $k^* \in [k,2k]$ such that 
the $k^*$-V-level in $\A(\Sigma^{e'})$ has $O(t\beta(t))$ vertices. This follows by the pigeonhole 
principle, since the overall number of vertices in the first $2k$ V-levels is $O(kt \beta(t/k)) = O(kt\beta(t))$ 
(a well-known consequence of the Clarkson-Shor argument, already used earlier in a different context).
We then clip the arcs of $\Sigma^{e'}$ at their intersections with the $k^*$-V-level of $\A(\Sigma^{e'})$, 
keeping only the subarcs that span smaller V-levels, creating $O(t\beta(t))$ arcs. Next, we 
further partition these clipped arcs into a total of $O(t\beta(t))$ smaller subarcs, so that each subarc contains at most 
$k$ vertices of $\A(\Sigma^{e'})$ (again, this is possible because, as just argued, the V-levels up to V-level $k^*$
contain a total of $O(kt\beta(t))$ vertices). 
% \micha{Again, this includes the surfaces from $\defn(\chi)$.}
Let $\ASet^{e'}$ denote the resulting set of subarcs, 
and consider their arrangement $\A(\ASet^{e'})$. Every vertical crossing of $\CSet^{e'}$ corresponds to a vertex
of $\A(\ASet^{e'})$ (because the bottom endpoint of the vertical segment of the crossing lies on two top 
boundaries, by construction, and is separated from $e'$ by at most $k$ surfaces).
%Somewhat abusing the notation,
We define the \emph{V-depth} of an arc $\xi\in \ASet^{e'}$ to be the smallest V-level in 
$\A(\ASet^{e'})$ of any point on $\xi$. By construction, the V-depth of $\xi$ is at most $2k$.
%and it is also equal to the smallest V-level in $\A(\Sigma^{e'})$ of any point on $\xi$
%\esther{this last sentence is a repetition of what is written in the definition of V-depth, keep just first part of this sentence?}. 

For $h \in [1,k^*]$, let $\ASet_h \subseteq \ASet^{e'}$ be the subset of arcs of V-depth $h$, and
let $t_h=|\ASet_h|$; we have $\sum_h t_h \le |\ASet^{e'}|=O(t\beta(t))$. Let 
$\chi=(e,e',g)$ be an uncovered vertical crossing in $\CSet^{e'}$ such that the V-level 
(in $V^{e'}$) of its bottom endpoint is $h$. 
As in the proof of~\cite[Lemma~2.5]{SS-97}, 
the crucial observation is that the bottom endpoint of $g$ is a vertex of the \emph{upper envelope} 
of $\ASet_h$. 
This is because the uncoveredness of $\chi$ means that, as we follow any of the two
arcs containing the bottom endpoint of $\chi$, say $\bd\sigma_1^+$, away from that endpoint, after each vertex 
that we cross, the V-level increases, for otherwise $\chi$ would have been covered. More precisely, the analysis 
applies to only top boundaries. Since the area just above $e$
is not contained in any region, any top surface that lies above $e$ has its bottom counterpart also above $e$.
Thus if a top surface $\bd\sigma^+$ crosses $\bd\sigma_1^+$, the corresponding bottom surface $\bd\sigma^-$ 
must have crossed it earlier, implying that $\chi$ is a covered crossing, contrary to assumption.
Hence, the number of uncovered vertical crossings in $\CSet^{e'}$ within $V^{e'}$ is at most
(where $s$ is the same parameter used to define $\beta(t)$ as $O(\lambda_{s+2}(t)/t)$)
\begin{equation}
	\sum_{h=1}^{k^*} \lambda_{s+2} (t_h) =
	\sum_{h=1}^{k^*} O(t_h\beta(t_h)) = 
	O(t\beta^2(t)).
\end{equation}
On the other hand, the number of vertices in $\A(\Sigma^{e'})$ whose V-level is at most $3k$ is $\Omega(tk)$. 
Indeed, at least $t$ arcs of $\Sigma^{e'}$ appear on the first $2k$ V-levels,
therefore at least $t-3k \ge t/2$ (recall that $t\ge 6k$) 
arcs of $\Sigma^{e'}$ contain a point at V-level larger than $3k$ and also appear on
a V-level smaller than $2k$ (within $V^{e'}$),
so these arcs contain a total of $\Omega(tk)$ vertices whose  
V-level is at most $3k$ in $\A(\Sigma^{e'})$.
Each such vertex $v$ is the bottom endpoint of a vertical crossing $\chi_v = (e_v,e',g_v)$ of $\A(\S)$, where 
$e_v$ is the edge of $\A(\S)$ containing $v$. Since $g_v$ intersects the top boundaries of at most $3k$ sets of 
$\S$ and the depth of $e'$ (with respect to any subset of $\S$) is at most 
$k$, the level of $\chi_v$  with respect to $\C(\S)$ is at most $4k$.
%\esther{Shall we complete these details to show how to obtain the bound $\frac{\beta^2 (n)}{k} \visib_{4k}(\S)$ in Lemma~\ref{lem:uncovered-crossings}?}

Summing over all (split) edges $e'$ of $\subklevel{k}(\S)$,
we obtain the following (see~\cite[Lemma~2.5]{SS-97} for a more detailed proof, in a different but similar context):
\begin{lemma}
  \label{lem:uncovered-crossings}
  The total number of uncovered vertical crossings in $\CSet$ is 
  \[
  O\left( k^3n^2\beta(n) + k^5\EE[\psi(\R)] + \frac{\beta^2 (n)}{k} \visib_{4k}(\S)\right ),
  \]
  where $\R$ is a random subset of $\S$ obtained by choosing each set of $\S$ with probability $1/k$,
and $\EE$ is the expectation with respect to this random choice.
\end{lemma}

%--------------------------------------------
\paragraph{Covered vertical crossings.}
Let $\CSet^c \subseteq \CSet$ be the set of all covered vertical crossings in $\CSet$. We prove that 
$|\CSet^c|=O(k^6\EE[\psi(\R)])$, where $\R$ is a random subset of $\S$ as defined above. Recall 
that each covered crossing $\chi=(e,e',g)\in\CSet^c$ involves some (at least one) set $\zeta\in\S$, 
such that (i) $\bd\zeta^-$ crosses the relative interior of $g$, 
(ii) $\bd\zeta^-$ crosses the (top) boundary of one of the sets, say, $\sigma_1$, that define the lower edge $e$, 
(iii) $\bd\zeta^-$ crosses $\bd\sigma(\chi)^-$ (before reaching $\xi$), and
(iv) there are at most $k$ vertices along $\bd\sigma(\chi)^-\cap V_e$ between the top endpoint of 
$g$ and the intersection of $\bd\zeta^-$ with $\bd\sigma(\chi)^-$, and at most $2k$ vertices along $\bd\sigma_1^+\cap
V^{e'}$ between the bottom endpoint of $g$ and the intersection of $\bd\zeta^-$ with $\bd\sigma_1^+$; 
see Figure~\ref{fig:covered_zeta_sigma}. 
Property (ii) follows from the definition of covered crossings, property (iii) follows from the construction of
$\Phi$ in the first stage, and property (iv) follows by the same construction and definition of covered crossings. 
Let $\cross(\chi)$ denote the set $\zeta$ that we use as the witness of $\chi$ being a covered crossing; 
if there are several such sets then choose arbitrarily one of them. 
Following the analysis of \cite{SS-97}, we bound $|\CSet^c|$ in two steps. First, we reduce the problem to 
bounding the number of covered crossings that are retained in a suitable random subset $\R$ of $\S$. 
Then we draw a path on $\bd\C(\R)$ for each crossing retained in $\R$, argue that the resulting 
system of paths forms a collection of planar graphs, possibly with digons (faces of degree $2$), and 
use a planarity argument, along with property (R3), to bound the number of retained crossings.
\smallskip

%\micha{In Fig. 6: in (a) it should be $\bd\sigma(\chi)^-$, no? And draw $\bd\zeta^-$ in a different color (blue?) in both parts.}
%\esther{Fixed. What was wrong with the green color of  $\bd\zeta^-$?}
\begin{figure}[htb]
  \begin{center}
  \begin{tabular}{ccc}
    \includegraphics[scale=0.4]{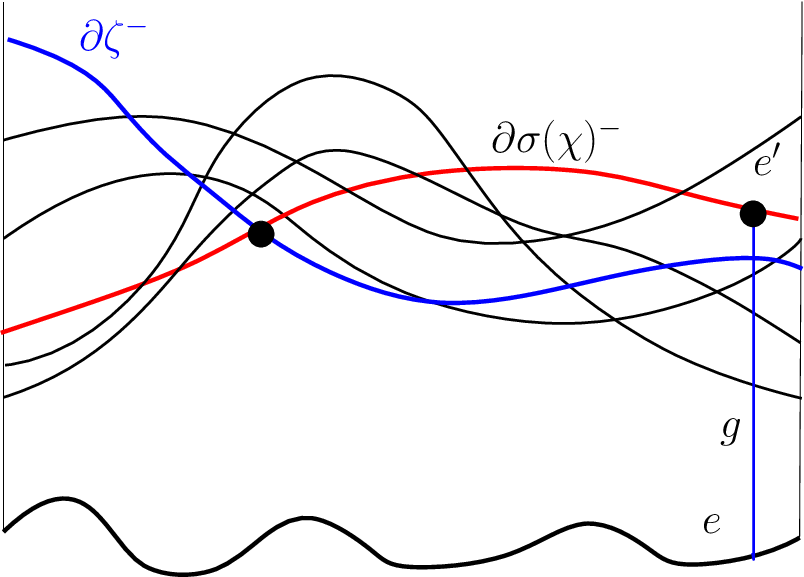} & \hspace*{1in} & 
    \includegraphics[scale=0.4]{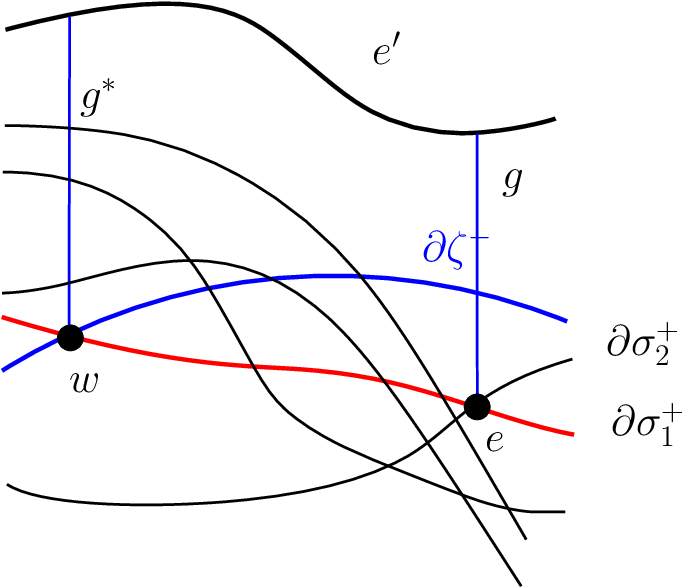}\\ %\hspace*{1.5in}\\
    (a) & \hspace*{1in} &(b)
  \end{tabular}
  \caption{The vertices (a) along $\bd\sigma(\chi)^-\cap V_e$ between the top endpoint of $g$ and the
intersection $\bd\zeta^- \cap \bd\sigma(\chi)^-$, and (b) along 
$\bd\sigma_1^+ \cap V^{e'}$ between $\chi=(e,e',g)$ and the intersection $w$ of $\bd\zeta^-$ and $\bd\sigma_1^+$.}
  \label{fig:covered_zeta_sigma}
  \end{center}
\end{figure}

%----------------------------
%\begin{lemma}
%$|\CSet^c| = O(k^6\psi(n/k))$.
%\end{lemma}
%----------------------------
%\begin{proof}
%
\textbf{\textit{Pruning and sampling.}}
We begin by pruning some of the vertical crossings from $\CSet^c$,
so that the following two properties hold for the pruned $\CSet^c$:
\begin{description}
\item[(P1)] For every subset $A\subset\S$ of size at most $4$, there is at most 
  one vertical crossing $\chi$ in $\CSet^c$ with $\defn(\chi) = A$ (recall that  $\defn(\chi)$ is the defining set of $\chi$),
  i.e., if there are multiple such crossings, we keep one and remove the (constantly many) others.
  %\esther{Shall we remind the notation of  $\defn(\chi)$?}
  See Figure~\ref{fig:prune}. %\esther{Fig label should be fixed to 7.}
\item[(P2)] For any three surfaces $\sigma_1, \sigma_2, \sigma_3 \in \S$, there is  at most one 
  crossing $\chi=(e,e',g)\in\CSet^c$ with $e\subset \bd\sigma_1^+\cap\bd\sigma_2^+$ 
  and $\sigma(\chi)=\sigma_3$ (again,  if there are multiple such crossings in the original 
  $\CSet^c$, we keep only one). 
\end{description}
Property (R3) and the fact that there are at most $s$ vertical crossings defined by the same four regions, 
for the previously defined parameter $s$, imply that
the pruning step reduces the size of $\CSet^c$ by only a factor of $O(k)$.
%Therefore it suffices to show that $|\CSet^c| = O(k^5\psi(\R))$ after this pruning step.

%\micha{In Fig. 7: Delete $\sigma_2^-$? Does not seem to belong.}
%\esther{Done. Also, change label to $\bd{\sigma_2^-}$?}
\begin{figure}[htb]
  \begin{center}
    \includegraphics[scale=0.4]{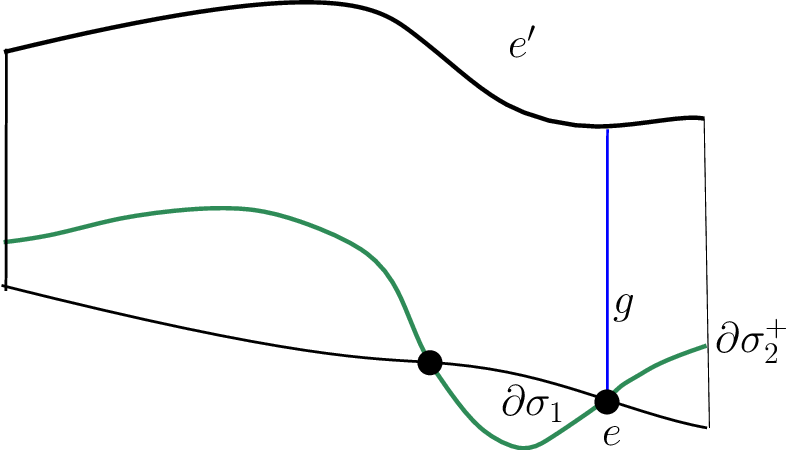}
    \caption{The two highlighted points represent vertical crossings, but we keep only the right one.}
    \label{fig:prune}
  \end{center}
\end{figure}

Let $\chi = (e,e',g)$ be a covered crossing in the pruned $\CSet^c$, and 
let $\triangle (\chi)$ (resp., $\triangle' (\chi)$) be the trapezoid within $V_e$ (resp., $V^{e'})$ introduced 
in property (R4) (resp., in the definition of covered crossings). By (R4) and this definition,
there are at most $k+2k=3k$ regions in $\S$ whose boundaries intersect $\triangle(\chi)\cup\triangle'(\chi)$. Denote this set of intersecting regions as $K(\chi)$. Since the interior of $\triangle'(\chi)$ may intersect one of the top surfaces incident on the edge $e$, $K(\chi)$ may contain one of the elements of $\defn(\chi)$. (See Figure~\ref{fig:covered}(b).)
We say that $\chi$ is \emph{retained} in a subset 
$\R \subseteq \S$ if (i) $\defn(\chi) \subseteq \R$, and (ii) 
$\R \cap (K(\chi)\setminus \defn(\chi)) =\{\cross(\chi)\}$.
If $\chi$ is retained in $\R$, then $\cross(\chi)$ is the only set in 
$\R$ whose boundary (bottom or top) intersects the relative interior of $\triangle(\chi)$ and thus of $g$. As 
mentioned above, besides $\bd\zeta^-$, one of the top surfaces incident on the lower edge $e$ of $\chi$ may intersect the 
interior $\triangle'(\chi)$ (and the corresponding bottom surface may also intersect $\triangle'(\chi)$)
but the boundary of no other set in $\R$ intersects the interior of that trapezoid.

As above, let $\R$ be a random subset of $\S$ obtained by choosing each set of $\S$ independently with probability $1/k$. 
Let $\CSet_\R^c \subseteq \CSet^c$ be the subset of covered crossings (after the pruning) that are retained in $\R$. 
Standard probabilistic arguments imply that a vertical crossing in $\CSet^c$ is retained in 
$\CSet_\R^c$ with probability at least $a_1/k^5$ for some constant $a_1>0$. Therefore we have
\begin{equation}
  \label{eq:covered_crossings}
  \EE[|\CSet_\R^c|] \ge \frac{a_1}{k^5} |\CSet^c|, \quad\text{or}\quad |\CSet^c| = O( k^5\EE[|\CSet_\R^c|])  
\end{equation}
(where expectation is with respect to the random choice of $\R$). 
Our main claim for covered crossings is that $|\CSet_\R^c| = O(\psi(\R))$, 
which would imply %the claimed inequality
$|\CSet^c| = O( k^5\EE[\psi(\R)])$. Once we show this, we obtain the asserted bound 
$|\CSet^c| = O( k^6\EE[\psi(\R)])$ for the original unprunned set $\CSet^c$. The proof proceeds as follows.

%\esther{Shall we label the latter equation, in order to refer it later in the analysis?}
%\micha{The inequality has $\psi(n/k)$. Need to justify why the expectation also satisfies this inequality.}
\smallskip

\textbf{\textit{Construction of the paths.}}
For every $\chi=(e,e',g) \in\CSet_\R^c$, we construct a path $\pi=\pi(\chi)$ on the 
boundary of the cells of $\C(\R)$, which is the concatenation of three subpaths
$\pi_1(\chi)$, $\pi_2(\chi)$, $\pi_3(\chi)$, defined as follows. Let $\sigma_1$, $\sigma_2$ be
the two regions such that $e\subseteq \bd\sigma_1^+\cap\bd\sigma_2^+$, and let $\zeta=\cross(\chi)$.
Let $v$ be the point on $\bd\sigma(\chi)^-$ at which $\bd\zeta^-$ crosses it for the first time (between 
the vertical visibility $\chi_0$ from which $\chi$ was charged and $\chi$, tracing $\bd\sigma(\chi)^-$ from 
$\chi_0$ to $\chi$). Let $u=g\cap\bd\zeta^-$, and let $w$ be the intersection point of $\bd\zeta^-$ with $\bd\sigma_1^+$
or $\bd\sigma_2^+$, say with $\bd\sigma_1^+$ for specificity,
that certifies $\chi$ as being a covered crossing (see the definition). We then define:
%\pankaj{Add the second case in the figure. Also add $\pi_1$, $\pi_2$, $\pi_3$ to the fig. Also, $\beta$ should be $\zeta$.}
%\esther{Done.}

%\micha{In Fig. 8: (i) Perhaps move $e$ further down along its curves (both parts). (ii) Why do we have two red points in (b)?}
%\esther{Fixed.}
\begin{figure}[htb]
  \begin{center}
    \begin{tabular}{ccc}
      \includegraphics[scale=0.4]{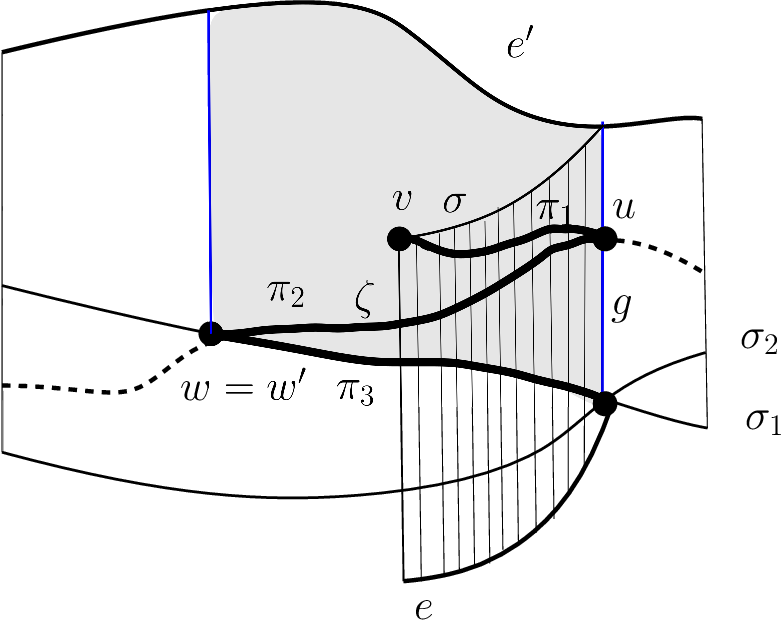} & \hspace*{0.5in} 
	    & \includegraphics[scale=0.4]{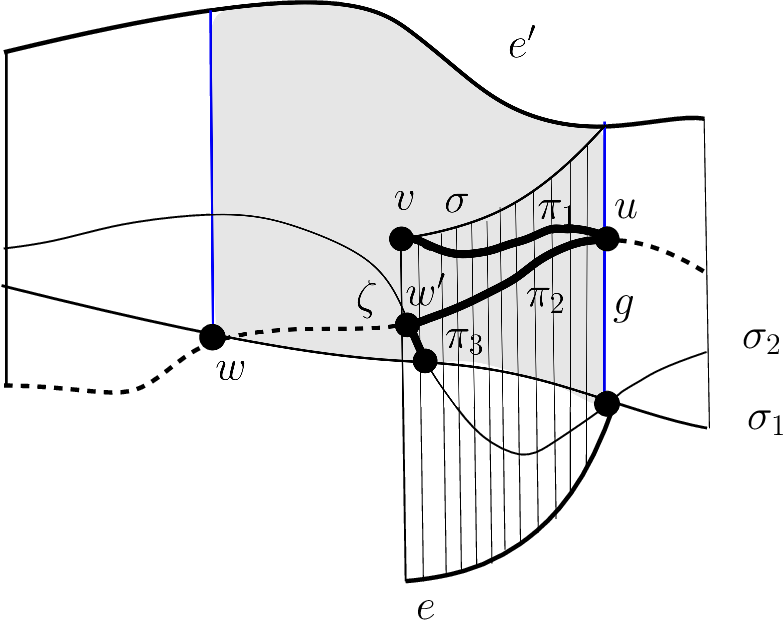}\\ %\hspace*{1in}\\
      (a)  &  & (b)
    \end{tabular}
    \caption{An illustration of a covered edge-crossing $(e,e',g)$ and its associated path 
      $\pi = \pi_1 \circ \pi_2 \circ \pi_3$ shown in bold. The two possible scenarios are depicted.  
    }
    \label{fig:path-2.6}
  \end{center}
\end{figure}

\begin{itemize}
\item[(i)] $\pi_1(\chi)$ is the subarc of $\bd\zeta^-\cap V_e$ connecting $v$ to $u$.
\item[(ii)] $\pi_2(\chi)$ is the subarc of $\bd\zeta^-\cap V^{e'}$ extending from $u$ 
  towards $w$, stopping as soon as it hits either $\bd\sigma_1^+$ or $\bd\sigma_2^+$ 
  at a point $w'$ ($w'$ may or may not be equal to $w$, depending on which of the
  two surfaces $\bd\sigma_1^+$, $\bd\sigma_2^+$ we encounter first; see the two cases in 
  Figure~\ref{fig:path-2.6}). 
\item[(iii)] $\pi_3(\chi)$ extends from $w'$ along its $e$-incident curve, say $\bd\sigma_1^+ \cap V^{e'}$,
  towards $g$, and stops as soon as it hits the second 
  $e$-incident surface $\bd\sigma_2^+$. 
  This terminal point may be either the bottom endpoint of $g$ 
  (see Figure~\ref{fig:path-2.6}(a)), or an earlier (more to the left) intersection point
  (see Figure~\ref{fig:path-2.6}(b)).
\end{itemize}

%As in \cite{SS-97}, we treat each boundary surface as two sided, 
%the first two subpaths $\pi_1(\chi), \pi_2(\chi)$ are drawn on the \emph{bottom} side of 
%$\bd\zeta^-$, and the third subpath $\pi_3(\chi)$ is drawn on the \emph{top} side of the respective surface 
%$\bd\sigma_1^+$ or $\bd\sigma_2^+$ (informally, they are drawn `outside' the respective regions of $\S$). Furthermore, $\pi_1(\chi) \subset \triangle(\chi)$ and $\pi_2(\chi) \circ \pi_3(\chi) \subset \triangle'(\chi)$.

Roughly, $\bd\zeta^-$ crosses one of the surfaces that form $e$ (within $V^{e'}$)
and one of the surfaces that form $e'$ (within $V_e$), at points that lie in the $2k$-neighborhood 
of $g$ in $\A(\S)$ along the respective surface-curtain intersection curves,
and are immediate neighbors of $g$ in $\C(\R)$. 
The path $\pi(\chi)$ walks along $\bd\zeta^-$ within the two respective curtains,
between $g$ and the intersection points $v$ and $w$ (or $w'$) of $\bd\zeta^-$
with these top and bottom surfaces (where, as just said, the bottom surface is incident with $e'$ and the top 
surface with $e$), and $\pi(\chi)$ also contains an additional connection along the $e$-incident top surface.
See Figure~\ref{fig:path-2.6}. Informally, drawing the path only along $\bd\zeta^-$ identifies, by the nature of its
endpoints, one bottom surface (among the pair forming $e'$) and one top surface (among the pair forming
$e$). This does not suffice to identify the edge-crossing $(e,e',g)$. The additional connection allows us to 
identify one more top $e$-incident surface, and thereby identify $\chi$, up to a factor of $k$ (before the pruning) 
or uniquely (after the pruning; again, recall property (R3)).
Note that each of $\pi_1(\chi)  \circ \pi_2(\chi)$ and $\pi_3(\chi)$ is fully contained (except for its endpoints) in a single face of $\C(\R)$. %\esther{Should be ``each of $\pi_1(\chi)$, $\pi_2(\chi)$ and $\pi_3(\chi)$ is fully contained...''?}
%\micha{Perhaps say that we draw the curves `on the outside' of the corresponding surfaces, so $\pi_1$ and $\pi_2$ are drawn on the lower face of the boundary, and $\pi_3$ is drawn on the upper face? Helps to visualize this, but not essential.}
%\esther{I Added the following sentence.}
%\esther{
We also emphasize that $\pi_1(\chi)$ and $\pi_2(\chi)$ are drawn on the ``bottom side'' of $\bd\zeta^-$,
and $\pi_3(\chi)$ is drawn on the ``top side'' of either $\bd\sigma_1^+$ or $\bd\sigma_2^+$. 
Furthermore, $\pi_1(\chi) \subset \triangle(\chi)$ and $\pi_2(\chi) \circ \pi_3(\chi) \subset \triangle'(\chi)$.

%The main technical step in~\cite{SS-97}, which we repeat and adapt here, is an argument that the 
\begin{lemma}
  \label{lem:path-disjoint}
  The paths $\pi(\chi)$, for $\chi\in \CSet_\R^c$,  are pairwise disjoint.
  %taking into account the sidedness of the boundary surfaces on which they are drawn.
\end{lemma}

%\micha{Not clear what Fig. 9(a) is trying to show in the context of type(a) intersections.}
%\esther{The two paths $\pi_1(\chi_1)$, $\pi_1(\chi_2)$ lie on the same curve, but have different endpoints $u_1 \neq u_2$.}

\begin{figure}[htb]
%  \vspace{1in}
  \begin{center}
    \begin{tabular}{ccc}
      \includegraphics[scale=0.5]{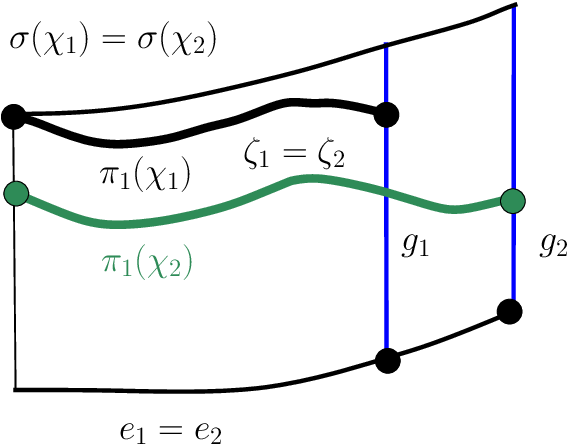}  & \hspace*{1in} & \includegraphics[scale=0.5]{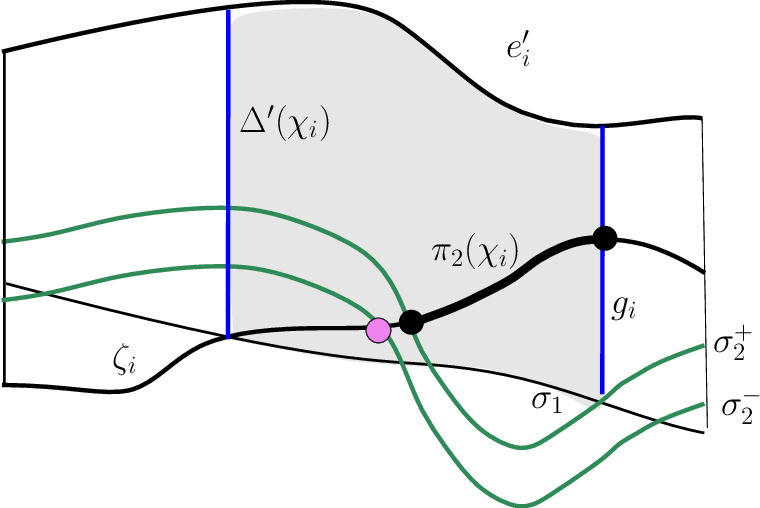}\\
      (a) & \hspace*{1in} & (b)
    \end{tabular}
  \end{center}
  \caption{Scenarios that arise in the analysis of the first two cases of path intersections.
  In (a) $\pi_1(\chi_1)$ and $\pi_1(\chi_2)$ are drawn on the same curve, and we separate them to aid the reader's visibility.}
  \label{fig:casei-ii}
\end{figure}

\begin{proof}
Let $\chi_1 = (e_1,e'_1,g_1)$ and $\chi_2 = (e_2,e'_2,g_2)$ be two distinct vertical crossings in 
$\CSet_\R^c$ (after the pruning), such that $\pi(\chi_1)\cap\pi(\chi_2)\ne\emptyset$. Let $q$ be such an intersection point. 
Let $\zeta_1 = \cross(\chi_1)$ and $\zeta_2 = \cross(\chi_2)$.
The way $\pi(\chi_1)$ and $\pi(\chi_2)$ are drawn,
$q$ can only be of one of the following four types (since $\pi_1(\chi_i), \pi_2(\chi_i)$ are drawn on 
bottom surfaces and $\pi_3(\chi_i)$ on top surfaces, $\pi_1(\chi_1), \pi_2(\chi_1)$ do not intersect $\pi_3(\chi_2)$ and 
	vice-versa):
\begin{description}
\item[(a)]
$q \in \pi_1(\chi_1) \cap \pi_1(\chi_2)$,
\item[(b)]
$q \in \pi_2(\chi_1) \cap \pi_2(\chi_2)$,
\item[(c)]
$q \in \pi_1(\chi_1) \cap \pi_2(\chi_2)$ or, symmetrically,
$q \in \pi_1(\chi_2) \cap \pi_2(\chi_1)$, or
\item[(d)]
$q \in \pi_3(\chi_1) \cap \pi_3(\chi_2)$.
\end{description}

\paragraph{Intersections of type (a).}
       Recall that $\bd\zeta_i^-$ is the only boundary surface that intersects the relative interior of 
       $\triangle(\chi_i)$ for $i=1,2$. It follows that if $q\in \pi_1(\chi_1)\cap \pi_1(\chi_2)$ then $e_1=e_2$, 
       i.e., the two subpaths are drawn on the same curtain; $\zeta_1=\zeta_2$; and the upper edges of
       $\triangle(\chi_1)$ and $\triangle(\chi_2)$ lie on the same bottom surface, namely the bottom surface that lies
       directly above $q$, i.e., $\sigma(\chi_1)=\sigma(\chi_2)$. See Figure~\ref{fig:casei-ii}~(a).
       Thus the two paths $\pi_1(\chi_1)$, $\pi_1(\chi_2)$ lie on the same curve.
         We show next that the endpoints $u_1 = g_1 \cap \zeta_1$ and $u_2 = g_2 \cap \zeta_2$ are equal, from which it will follow that $\chi_1 = \chi_2$.
       If $g_1$ and $g_2$ lie on the same vertical line, then obviously $\chi_1=\chi_2$, so, without 
       loss of generality, assume that $g_1$ lies to the left of $g_2$. By construction, $e'_1$ 
       lies on the intersection curve of $\bd\sigma(\chi_1)^-$ and another bottom surface 
       $\bd\eta_1^-\ne \bd\zeta_1^-=\bd\zeta_2^-$. Then $\bd\eta_1^-$ lies below $\bd\sigma(\chi_1)^-$ 
       immediately to the right of $g_1$ and thus intersects the relative interior of $\triangle(\chi_2)$. 
       This, however, contradicts the fact that $\bd\zeta_2^-$ is the only bottom surface that intersects 
       the relative interior of $\triangle(\chi_2)$. Hence, $g_1=g_2$ and thus $\chi_1=\chi_2$.

\paragraph{Intersections of type (b).}
        For $i=1,2$, $\bd\zeta_i^-$ and possibly one of the surfaces $\bd\sigma^-_1$ or $\bd\sigma^-_2$,
        whose top counterparts are incident on the bottom endpoint of $g_i$, are
	the only bottom boundary surfaces that intersect the relative interior of $\triangle'(\chi_i)$ 
        (in the random sample $\R$). Concretely, if $\triangle'(\chi_i)$ is bounded on its lower side 
        by $\bd\sigma_1^+$, and $\bd\sigma^+_2$ penetrates into $\triangle'(\chi_i)$, as in 
        Figure~\ref{fig:casei-ii}(b), then it is possible for $\bd\sigma_2^-$ to also penetrate 
        $\triangle'(\chi_i)$. But then, as is easily checked, $\bd\sigma_2^-$ cannot meet $\bd\zeta_i^-$ 
        at the portion on which $\pi_2(\chi_i)$ is drawn. Hence, if $\pi_2(\chi_1) \cap \pi_2(\chi_2)$ 
        is non-empty then $\bd\zeta_1^- = \bd\zeta_2^-$ and the intersection point $q$ lies below $e_1'$ 
        and $e_2'$, implying that $e'_1=e'_2$, i.e., $\pi_2(\chi_1)$ and $\pi_2(\chi_2)$ are drawn on 
        the same curtain, and on the same surface $\bd\zeta_1^- = \bd\zeta_2^-$. This in turn implies 
        that $\pi_2(\chi_1)$ and $\pi_2(\chi_2)$ overlap and one of them contains the endpoint of the other. 
	By construction, for any vertical crossing $\chi=(e, e',g)$, $\pi_2(\chi)$ starts vertically above $e$
	and ends as soon as it meets another edge of $\A(\R)$ (namely, the edge formed by the intersection 
        of $\cross(\chi)$ and a top surface incident on $e$; see Figure~\ref{fig:casei-ii}~(b)).
        Thus the relative interior of $\pi_2(\chi)$ does not intersect any edge of $\A(\R)$,
        implying that the left endpoints of $\pi_2(\chi_1)$ and $\pi_2(\chi_2)$ are the same. 
        Furthermore, the only top surfaces that may intersect the relative interior of 
        $\triangle'(\chi)$ below $\cross(\chi)$ are the ones incident on the lower edge $e$ of $\chi$,
	concluding that %the left endpoints of $\pi_2(\chi_1)$ and $\pi_2(\chi_2)$ are the same, and that
        $e_1$ and $e_2$ lie on the intersection curve of the same pair of top surfaces, i.e.,
	$\defn(\chi_1)=\defn(\chi_2)$. This, however, contradicts property (P1) of $\CSet^c$ (after pruning). 

\begin{figure}[htb]
	\begin{center}
	\begin{tabular}{ccc}
	\includegraphics[scale=0.4]{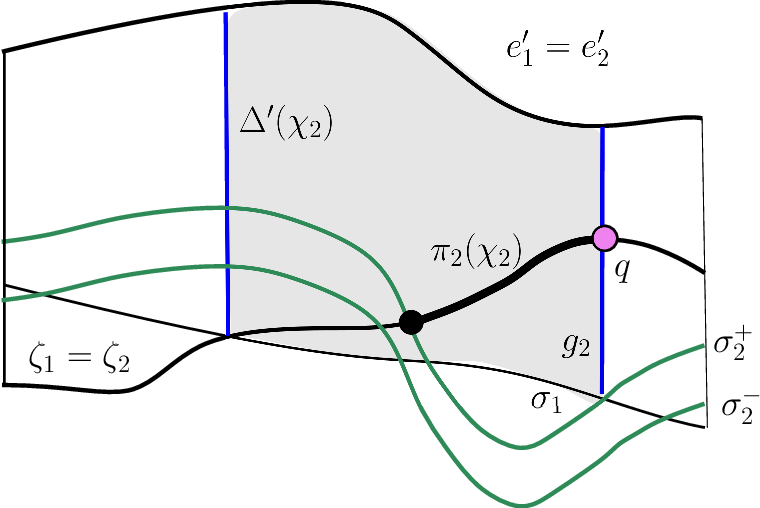} & \hspace*{0.5in}& \includegraphics[scale=0.4]{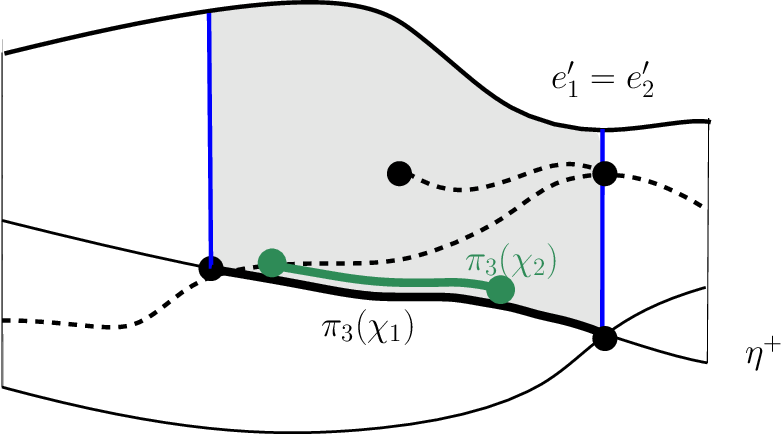}\\
	(a) & & (b)
	\end{tabular}
	\end{center}
	\caption{Path intersections of types (c) and (d).}
	\label{fig:caseiii-iv}
\end{figure}

\paragraph{Intersections of type (c).}
In this case, $q \in V_{e_1} \cap V^{e'_2}$. Thus $q \in g^*$ for some vertical crossing
$\chi^* = (e_1,e'_2,g^*)$ (which is not necessarily a crossing in $\CSet$). Since $\bd\zeta_i^-$ is the only 
bottom surface that intersects the relative interior of $\triangle(\chi_1)$ (in the sample $\R$), $\zeta_1=\zeta_2$.
Moreover, the only bottom surface  that might cross $\triangle'(\chi_2)$ is the bottom counterpart of one 
of the top surfaces that are incident on the bottom endpoint of $g_2$.
However, as argued in case~(b),  that bottom surface cannot cross $\pi_2(\chi_2)$, or pass above it,
implying that $e'_1=e'_2$. See Figure~\ref{fig:caseiii-iv}~(a).
Furthermore, no top surface intersects $\triangle(\chi_1)$ below $\bd\zeta_1^-$, and the only 
top surfaces that may intersect $\triangle'(\chi_2)$ below $\bd\zeta_2^- = \bd\zeta_1^-$ are the surfaces that 
are incident on the lower edge $e_2$ of $\chi_2$. We can conclude that $e_1$ and $e_2$ lie  on the intersection curve
of the same pair of top surfaces,
and thus $\defn(\chi_1)=\defn(\chi_2)$, which, as above, is a contradiction.

\paragraph{Intersections of type (d).}
By construction, $\pi_3(\chi_i)$ lies below an edge of $\A(\R)$ (the one containing $e'_i$), and the bottom 
surface of exactly one set of $\R$, namely $\bd\zeta_i^-$, lies between the path and the edge.
Thus, the existence of a point $q \in \pi_3(\chi_1)\cap\pi_3(\chi_2)$ implies that (i) $e'_1=e'_2$;
(ii) both $\pi_3(\chi_1)$ and $\pi_3(\chi_2)$  lie on the same curtain, namely $V^{e'_1}=V^{e'_2}$; and (iii)
both of these paths are drawn on the same top surface, say, $\bd\eta^+$.
%\esther{Important: Why do we have the same surface (either  $\bd\sigma_1^+$ or $\bd\sigma_2^+$)?  We think the reason is the following: If $\pi_3(\chi_1)$ and $\pi_3(\chi_2)$ were sitting on two different surfaces, $\bd\sigma_1^+$ and $\bd\sigma_2^+$, then the two subpaths could never intersect each other, as by definition we draw $\pi_3$ on $\bd\sigma_2^+$ and stop as soon as we hit $\bd\sigma_1^+$. So they have to sit on the same surface.}
See Figure~\ref{fig:caseiii-iv}~(b).
Since the endpoints of $\pi_3(\chi_i)$ are two consecutive vertices of 
$\A(\{ \bd\sigma^+ \cap V^{e'_i} \mid \sigma \in \R\})$
(i.e., $\pi_3(\chi_i)$ is an edge of this cross-sectional two-dimensional arrangement), the existence of $q$ implies that 
$\pi_3(\chi_1)=\pi_3(\chi_2)$. Since $e_i$ lies on the intersection curve of $\bd\eta^+$ and the top surface that
is incident on the right endpoint of $\pi_3(\chi_i)$,
we conclude that both $e_1$ and $e_2$ lie on the intersection curve of the same pair of top surfaces, and 
thus $\defn(\chi_1)=\defn(\chi_2)$, which again leads to a contradiction.

Hence, we conclude that the relative interiors of any pair of
distinct paths $\pi(\chi_1)$ and $\pi(\chi_2)$ are disjoint, as claimed.
\end{proof}

We thus obtain a system of pairwise-disjoint paths drawn on the boundaries of cells of $\C(\R)$. 
Each path $\pi(\chi)$ is a concatenation of its two subpaths 
$\pi_1(\chi) \circ \pi_2(\chi)$ and $\pi_3(\chi)$, each of which is drawn on a single face of $\bd \C(\R)$ and connects two 
edges of that face. By planarity, the number of subpaths within a single face is proportional to the complexity of 
that face plus the number of digons in the drawing of that graph within that face.
(Technically, a digon is really drawn as a quadrangular region, bounded by two subpaths and by two
portions of boundary arcs of the face that are connected by both subpaths.)
Hence the overall number of subpaths is $O(\psi(\R))$ plus a term proportional to the overall
number of digons. Using property (P2) of $\CSet^c$ and following exactly the same argument
as in \cite{SS-97}, it can be shown that the number of digons is also $O(\psi(\R))$.
We refer to~\cite{SS-97} for details.
Hence, $|\CSet_\R^c| = O(\psi(\R))$. We thus obtain the following (see once again~(\ref{eq:covered_crossings}) and the discussion around it):

%----------------
\begin{lemma}
  \label{lem:covered-crossings}
  The overall number of covered edge-crossings in $\CSet$ is $O(k^6 \EE[\psi(\R)])$, where $\R$ is a random subset of 
  $\S$ obtained by choosing each set of $\S$ independently with probability $1/k$. 
\end{lemma}
%----------------

\paragraph{Putting all the pieces together.}
Combining Lemmas~\ref{lem:uncovered-crossings} and~\ref{lem:covered-crossings}, the total number of vertical crossings in $\CSet$ is
\begin{equation}
  \label{eq:R2}
  |\CSet| = O\left( k^6\EE[\psi(\R)]) + k^3n^2\beta(n) + \frac{\beta^2(n)}{k}\cdot \visib_{4k}(\S) \right) .
\end{equation}

Substituting (\ref{eq:R2}) in (\ref{eq:R}) and using the Clarkson-Shor analysis technique to bound 
$\visib_{4k} (\S)$ by $O(k^4 \EE[\visib_0(\R)])$ (recalling that each element counted in 
$\visib_{4k} (\S)$ is defined by four surfaces), we obtain the final recurrence
\begin{equation}
  \label{eq:final}
  \visib_0(\S) = O\Bigl( k^5\beta(n) \EE[\psi(\R)] + k^2\beta^2(n)(n^2+\psi(\S)) 
+ k^2\beta^3(n)\cdot \EE[\visib_0(\R)] \Bigr) ,
\end{equation}
where $\R$ (in both the first and last terms) is a random subset of $\S$ as above.\footnote{%
  We actually use a separate sample for each part of the analysis, but we might as well use the same sample twice.} 
Using Chernoff's bound, one can show that 
$$
\Pr[|\R| \ge 2n/k] \le \exp(-n/3k).
$$
Since $\psi(n)=O(n^3)$, we obtain $\EE[\psi(\R)] = O(\psi(n/k))$ (since the part of the expectation contributed by samples with
$|\R| > 2n/k$ is negligible). 
For $\psi(n) = \Omega(n^2)$, $\psi(n/k) = O( \psi(n)/k^2)$.
Similarly, we can argue that $\EE[\visib_0(\R)] = O(\visib_0(n/k))$.
Therefore (\ref{eq:final}) can be rewritten as
\begin{equation}
	\label{eq:visib-rec}
	\visib_0(n)  \le c_1 k^3 \beta^2(n) \left( n^2 + \psi(n) \right) + c_2k^2\beta^3(n)\cdot \visib_0(n/k) , 
\end{equation}
where $c_1, c_2>0$ are constants.

We solve (\ref{eq:visib-rec}) in a more explicit manner, before resorting at the end back to our
$O^*(\cdot)$ notation. Specifically, let $\eps>0$ be an arbitrarily small constant.
We choose the parameter $k$, as a function of $n$, such that\footnote{%
  Recall our earlier comment that $k$ needs to be chosen as a non-constant parameter.}
$$k \ge (4c_2\beta^3(n))^{1/\eps}.$$ 
We claim that, with this choice of $k$, the solution of the recurrence in~(\ref{eq:visib-rec}) is 
\begin{equation} 
	\label{eq:recur-sol}
\visib_0(n) \le A k^3 \beta^2(n) (n^2+\psi(n)) \cdot n^\eps ,
\end{equation}
for a sufficiently large constant $A \ge 2c_1$ that also depends on $\eps$. 
If $n\le n_0$ for some constant $n_0$ that depends on $\eps$, 
then (\ref{eq:recur-sol}) holds by choosing $A$ sufficiently large. If $n>n_0$, we use induction, and plug
(\ref{eq:recur-sol}), with $n/k$ replacing $n$, into (\ref{eq:visib-rec}), to obtain 

\begin{align*}
	\visib_0(n)  & \le c_1 k^3 \beta^2(n) \left( n^2 + \psi(n) \right) + 
	c_2k^2\beta^3(n) \cdot A k^3 \beta^2(n/k) \left [ \biggl(\frac{n}{k} \biggr)^2+
	\psi \biggl(\frac{n}{k} \biggr )\right ] \biggl(\frac{n}{k} \biggr )^\eps .
\end{align*}
As already noted, it can be shown, with some care, that
\[ 
\biggl(\frac{n}{k} \biggr)^2+ \psi \biggl(\frac{n}{k} \biggr ) \le
\frac{2}{k^2} \left (n^2+ \psi (n)\right) .
\]
Plugging this into the above inequality, we obtain
\begin{align*}
  \visib_0(n)
  &\le c_1 k^3 \beta^2(n) \left( n^2 + \psi(n) \right) + 
  2c_2\beta^3(n) \cdot Ak^3\beta^2(n) (n^2+\psi(n)) \biggl(\frac nk \biggr )^\eps \\
  & \le A k^3\beta^2(n) (n^2+\psi(n)) n^\eps \left [ \frac{c_1}{An^\eps} + 
    \frac{2c_2\beta^3(n)}{k^\eps} \right ] \\
  & \le A k^3\beta^2(n) (n^2+\psi(n)) n^\eps ,
\end{align*}
where the last inequality follows from the choices of $A$ and $k$.
Hence, returning to our usual notation,
$\visib_0(n) = O^*(n^2+\psi(n))$, which finally completes the proof of Theorem~\ref{thm:main}.
$\Box$

%-----------------------------------------------------------
\subsection{Output-sensitive bound for the vertical decomposition of the entire arrangement}
\label{subsec:output_sensitive}

We next extend the above analysis to obtain an output-sensitive bound on $\VeD{\A}(\S)$, 
the vertical decomposition of the entire $\A(\S)$. In particular, we prove that the complexity 
of $\VeD{\A}(\S)$ is at most $O^*(n^{2})+c\beta(n)^{O(1/\eps)}\cdot\mu $, 
where $\mu := \mu(\S)$ is the number of vertices in $\A(\S)$, $\eps>0$ is  any constant,
$c$ is a suitable constant, 
and $\beta(n)=\lambda_s(n)/n$ is an extremely slowly growing function.
For any subset $\R\subseteq \S$, let $\psi(\R)$ and $\visib_0(\R)$ denote the complexity of $\A(\R)$ and 
the number of vertical crossings of level $0$ (i.e., vertical visibilities) in $\A(\R)$, respectively.
As mentioned in the beginning of the section, 
$\psi(\R)=O(|\R|^2+\mu(\R))$, where $\mu(\R)$ is the number of vertices of $\A(\R)$.
The analysis in the previous subsection can be adapted (in fact, somewhat simplified)  to the vertical decomposition of the 
entire arrangement, but instead of repeating the argument, we reduce the problem of bounding the complexity of 
$\VeD{\A}(\S)$ to the case of the complement of the union, as follows. As above, assume that each face of the
boundary of a region in $\S$ is $xy$-monotone and does not contain any singular point in its relative interior.
Let $\Sigma$ be the resulting set of $xy$-monotone surface patches; $\A(\S)$ is the same as $\A(\Sigma)$,
except for additional features caused by the splitting of surfaces into subsurfaces, whose
number is easy to estimate, and is subsumed by the overall bound that we get.
We treat each surface patch $\sigma$ as an infinitesimally thin region. More specifically, 
for $\delta>0$, let $B_\delta \subset \reals^3$ be the ball of radius $\delta$ centered at the origin. Let
$\sigma_\delta :=\sigma\oplus B_\delta$ be the Minkowski sum of $\sigma$ with $B_\delta$. 
Set $\Sigma_\delta = \{\sigma_\delta \mid \sigma\in \Sigma\}$. We view $\sigma$ as 
$\lim_{\delta\rightarrow 0} \sigma_\delta$. Each full-dimensional
cell of $\A(\Sigma)$ can be identified with a cell of $\lim_{\delta\rightarrow 0} \C(\Sigma_\delta)$. Hence,
the complexity of $\VeD{\A}(\Sigma)$ is bounded by that of $\VeD{\C}(\Sigma_\delta)$, 
for some sufficiently small $\delta$.
(We also recall our general position assumption made at the beginning of this section.)
%\micha{Do we need to assume here general position?}
%\esther{We assume general position at the beginning of Section~\ref{sec:union}, then I assume we still need to assume it here (to have a unique quadruple that defines a vertical crossing).}
We thus use (\ref{eq:final}) to bound $\visib_0(\S)$. 
For $|\S|\ge n_0$, where $n_0$ is a sufficiently large constant,
we obtain the following recurrence for $\visib_0(\S)$:
\begin{equation}
  \label{eq:final-full}
  \visib_0(\S) = O\Bigl(k^5\beta(n)\EE[\psi(\R)] + k^2\beta^2(n)(n^2+\psi(\S)) +
  k^2\beta^3(n)\cdot \EE[\visib_0(\R)] \Bigr) .
\end{equation}
As above, $\R$ is a random subset of $\S$ in which each element of $\S$ is chosen with probability $1/k$,
and $\EE[\cdot]$ denotes expectation with respect to this choice.  Let $n_{\R}$ be the random variable 
that denotes the number of sets in $\R$, let $\mu_{\R}$ denote the random variable that denotes the 
number of vertices in $\A(\R)$, and let $\visib_0(n,\mu)$ denote the maximum number of vertical 
visibilities in $\A(\S)$ with $|\S| \le n$ and $\mu(\S) \le \mu$.
%As in the previous section, we obtain that 
%$$
%\Pr[n_R \ge 2n/k] \le \exp(-n/3k),
%$$ 
%and thus
We now observe that, for $n \ge n_0$,
\begin{equation}
  \label{eq:exp-bound}
  \EE[n_{\R}^{\alpha}] \le 3\left ( \frac{2n}{k}\right)^{\alpha} \;\; \mbox{for $\alpha\le 3$},
  \quad \mbox{and} \quad \EE[\mu_{\R}] \le a_1 \left(\left(\frac{n}{k}\right)^2+\frac{\mu}{k^3} \right) ,
\end{equation}
for some constant $a_1 \ge 1$, where the first inequality follows from the property (implied by
Chernoff's bound, and also stated in the previous section): 
$$
\Pr[n_{\R} \ge 2n/k] \le \exp(-n/3k) .
$$ 
Therefore 
$$
\EE[\psi(\R)] = O(\EE[n_{\R}^2]+\EE[\mu_{\R}]) = O((n/k)^2+\mu/k^3).
$$
Plugging this bound into (\ref{eq:final-full}), we obtain the following recurrence
\begin{equation}
  \label{eq:upper-bound}
  \visib_0(n,\mu) \le c_1 k^3\beta^2(n)(n^2+\mu)+  c_2 k^2\beta^3(n)\cdot \EE[\visib_0(n_{\R},\mu_{\R})],
\end{equation}
where $c_1, c_2>0$ are constants.

As above, let $\eps>0$ be an arbitrarily small constant. We choose 
$$
k := k(n) = a_1(96c_2\beta^3(n))^{1/\eps} ,
$$
and claim that, with this choice,
the solution to the recurrence in~(\ref{eq:upper-bound}) is:
\[
\visib_0 (n,\mu) \le A k(n)^3 \beta^2(n)(n^{2+\eps} + \mu) ,
\]
where $A > 0$ is a sufficiently large constant that depends on $\eps$ and the other parameters. Indeed, 
the case $n\le n_0(\eps)$, for a suitable constant threshold $n_0(\eps)$, is easy.
For larger values of $n$,
using induction and observing that $k(n_{\R}) \le k(n)$ and $\beta(n_{\R}) \le \beta(n)$, we have
\begin{align*}
	\visib_0(n,\mu) & \le  c_1 k^3\beta^2(n)(n^2+\mu)+  c_2 k^2\beta^3(n)\cdot 
			\EE \left [A k^3\beta^2(n)(n_{\R}^{2+\eps}+\mu_{\R}) \right ]\\
	 & \le  c_1 k^3\beta^2(n)(n^2+\mu)+  c_2 k^2\beta^3(n)\cdot 
			A k^3\beta^2(n) \left ( \EE[n_{\R}^{2+\eps}]+ \EE[\mu_{\R}]\right )\\
	 & \le  c_1 k^3\beta^2(n)(n^2+\mu)+  c_2 k^2\beta^3(n)\cdot 
			A k^3\beta^2(n) \left (3\left ( \frac{2n}{k}\right)^{2+\eps} + 
			a_1 \left(\left(\frac{n}{k}\right)^2+\frac{\mu}{k^3} \right)\right )\\
			&\qquad\mbox{(Using (\ref{eq:exp-bound}))}\\
			&\le Ak^3\beta^2(n) \left [ n^{2+\eps} \left (\frac{c_1}{An^\eps}+\frac{a_1c_2\beta^3(n)}{n^\eps}+ 
			\frac{24c_2\beta^3(n)}{k^\eps}\right ) + 
			\mu \left ( \frac{c_1}{A}+\frac{a_1c_2\beta^3(n)}{k} \right ) \right ]. %\\
\end{align*}          
%&\qquad\mbox{(
We choose $\eps\le 1/2$,  $A \ge 2c_1$,
$n_0(\eps)$ sufficiently large so that for $n \ge n_0(\eps)$ we have $n^\eps > 4a_1c_2\beta^3(n)$,
and observe that $k^\eps = 96a_1^{\eps}c_2 \beta^3(n)$. It then follows that the latter expression is at most:
\[
%&\le
Ak^3\beta^2(n) ( n^{2+\eps} + \mu) .
\]
%\end{align*}
This completes the induction step and thus implies:

 %------------------------------------------
 \begin{theorem}
   \label{thm:arrg}
   Let $\S$ be a collection of $n$ constant-complexity semi-algebraic sets
   in $\reals^3$, and let $\mu$ be the number of vertices in $\A(\S)$. Then 
   the complexity of the vertical decomposition of $\A(\S)$ is 
   $O\Bigl(n^{2+\eps} + \beta(n)^{c/\eps}\mu\Bigr)$, for any constant $\eps>0$ and a suitable constant $c> 0$, 
   where $\beta(n)=\lambda_s(n)/n$ is an extremely slowly growing function.
   That is, the complexity of the decomposition is $O^*(n^2+\mu)$.
\end{theorem}
%------------------------------------------

 \noindent{\bf Remark:}
 We note that the output-sensitive bound shown in Theorem~\ref{thm:arrg} carries an $n^{\eps}$-artifact.

%-----------------------------------------------------
\section{Vertical Decomposition in $\reals^4$}
\label{sec:env}

In this section we extend the analysis of the previous section to the case of 
vertical decomposition of lower envelopes and
of the entire arrangement in $\reals^4$. We reduce the problem to bounding the 
total size of the vertical decompositions of many instances of substructures of 3D 
arrangements. Although the complexity of the vertical decomposition of one instance might be large, 
we argue that the total size, summed over all instances, is $O^*(n^3)$ for lower envelopes and 
$O^*(n^4)$ for the entire arrangement. We begin with the case of lower envelopes, which is somewhat simpler.

%--------------------------------------
\subsection{Lower envelopes}
\label{subsec:lower_env}

Let $\F$ be a collection of $n$ trivariate semi-algebraic functions of constant complexity,
let $E = E_\F$ denote the lower envelope of $\F$, let $E^- = E^-_\F$ denote the portion of $\reals^4$ below $E$, 
and let $M = M_\F$ denote the minimization diagram of $E$, namely the projection of $E$ onto the
$xyz$-space. Note that $M$ forms a partition of ${\reals}^3$ into cells of arbitrary complexity (but this partition is not an arrangement of semi-algebraic surfaces).
Our goal is to estimate the combinatorial complexity of the vertical decomposition of $M$.
This three-dimensional decomposition can then be lifted up in the $w$-direction to induce a suitable
decomposition of $E^-$, or a suitable decomposition of $E$ itself; we refer to either of these 
decompositions as the vertical decomposition of $E$. We note that the complexity of (the undecomposed) 
$E$ and of $M$ is $O^*(n^3)$~\cite{SA}. The main result of this subsection yields the same asymptotic 
bound, up to the $O^*(\cdot)$ notation, for their vertical decomposition:
%------------------------------------------
\begin{theorem}
  \label{thm:env4d}
  The complexity of the vertical decomposition of the lower envelope (that is, of the minimization diagram) 
  of a collection of $n$ trivariate semi-algebraic functions of constant complexity is $O^*(n^3)$.
\end{theorem}
%------------------------------------------

\noindent{\bf Proof.}
We assume that the functions of $\F$ are in general position, continuous and totally defined.
None of these assumptions are essential, and well known standard techniques can be used to handle 
setups where these assumptions fail, but these assumptions simplify the analysis. 
We identify each function of $\F$ with its three-dimensional graph. 
We recall the way in which the vertical decomposition $\VeD{M}$ of $M$ is constructed. 
We fix a function $a$ in $\F$. For each function $b\in\F\setminus\{a\}$, 
we use $\sigma_{ab} = \sigma_{ba}$ to denote the $xyz$-projection of the two-dimensional intersection 
surface $a\cap b$. The surface $\sigma_{ab}$ partitions the $xyz$-space into the regions $\sigma_{ab}^+$ and
$\sigma_{ab}^-$, where $\sigma_{ab}^+$ (resp., $\sigma_{ab}^-$) consists or those points $(x,y,z)$
for which $a(x,y,z) \ge b(x,y,z)$ (resp., $a(x,y,z) \le b(x,y,z)$). We observe that the complement
$\C_a$ of the union $\U_a := \bigcup \left\{ \sigma_{ab}^+ \mid b \in\F\setminus\{a\} \right\}$ 
(which is the same as the intersection $\bigcap \left\{ \sigma_{ab}^- \mid b \in\F\setminus\{a\} \right\}$)
is precisely the portion of the $xyz$-space over which $a$ attains the envelope $E$. We denote the collection
$\left\{ \sigma_{ab}^+ \mid b \in\F\setminus\{a\} \right\}$ as $\Sigma^+_a$.

We now construct the three-dimensional vertical decomposition $\VeD{\C_a}$ of $\C_a$, and repeat this construction 
to each complement $\C_a$, over $a\in\F$, observing that the regions $\C_a$ are pairwise openly disjoint. The union 
of all these decompositions yields the vertical decomposition of $M_\F$, and, as mentioned above, the vertical
decomposition of $E_\F$ is obtained by lifting this decomposition to $E_\F$ (or to $E_\F^-$, see above), in 
a straightforward manner. Concretely, for the decomposition of $E_\F^-$, each cell $\tau$ in the decomposition of $M$ is lifted to the semi-unbounded region
\[
\{ (x,y,z,w) \mid (x,y,z)\in\tau \;\text{and}\; w\le E(x,y,z) = f_\tau(x,y,z) \} ,
\]
where $f_\tau$ is the unique function of $\F$ that attains $E$ over $\tau$. (For the decomposition of $E_\F$,
replace the inequality in the above expression by an equality.) We have thus (almost) reduced the 
problem to the problem studied in Section~\ref{sec:union}. The difference is that there we had a
uniform bound on the complexity of the union of any subcollection of at most $m$ of the given regions, 
which was assumed to be $O^*(m^2)$. Here, in contrast, this no longer holds. That is, considering 
the entire collection $\F$, and denoting by $N_a$ the complexity of $\U_a$, all we know is that 
$\sum_a N_a = O^*(n^3)$, so we have the bound $O^*(n^2)$ only for the \emph{average} value of $N_a$. As the analysis
below shows, this is not a major issue, but it means that we cannot use the analysis of Section~\ref{sec:union}
as a black box. (Informally, this is because we have no control on this complexity for subsets of regions,
if we construct each sample independently.)
Instead, we apply it in a somewhat modified manner, which we spell out next.

We apply, within each $\C_a$, a suitably adjusted version of the analysis of Section~\ref{sec:union}, making occasional 
references, as needed, to the original analysis of \cite{SS-97}, and to the analysis of Section~\ref{sec:union}.

The first charging step constructs, for each $a\in\F$, a set $\CSet_a$ of $z$-vertical edge-crossings,
with respect to $\C_a$, using the same machinery as in Section~\ref{sec:union} (and in \cite{SS-97}).
Referring to~(\ref{eq:R}), we obtain the modified estimate (where here we denote by $\visib_0(\Sigma_a^+)$ the number of $z$-vertical visibilities within $\C_a$):
\begin{equation}
  \label{eq:Ra4d}
 % N_0(\Sigma_a^+)
  \visib_0(\Sigma_a^+)
  \le  \frac{c \beta(n)}{k} |\CSet_a| +
  O^*\left( k\beta(k) \left( |\C_a| + k\beta(n/k)n^2 
  \right) \right) ,
\end{equation}
where $c > 0$ is an appropriate constant.
As before, the overhead expression bounds the number of $z$-vertical visibilities that involve the right portions 
of the edges of $\C_a$, where the first term accounts for the original unsplit edges, and the second term accounts for
the number of splits.
Summing up the inequalities (\ref{eq:Ra4d}) over all $a\in\F$, we get (recalling that $|M| = O^*(n^3)$):
\begin{equation}
  \label{eq:R4d}
  \begin{aligned}
  %N_0(\F)
    \visib_0(\F) & \le
    \frac{c\beta(n)}{k} |\CSet| 
    + O^*\left( k\beta(k) \left( |M| + k\beta(n/k)n^3 
    \right) \right) \\
    & =
    \frac{c\beta(n)}{k} |\CSet| 
    + O^*\left( k^2 \beta^2(n) n^3 \right) ,
  \end{aligned}
\end{equation}
where $\visib_0(\F)$ is the overall number of $z$-vertical visibilities in $M$, and $\CSet = \bigcup_{a\in\F} \CSet_a$.

The second charging step is governed by 
Lemmas~\ref{lem:uncovered-crossings} and~\ref{lem:covered-crossings}
of Section~\ref{sec:union} 
(which are analogous to Lemmas 2.5 and 2.6 in \cite{SS-97}),
which bound, respectively, the number of uncovered and covered edge-crossings. The proof of \lemref{uncovered-crossings} 
is applied, roughly unchanged, separately for each $a\in\F$, with respect to the complement $\C_a$ and the 
collection $\Sigma_a^+$, with a suitable interpretation of the edges $e'$ and the quantities $t_{e'}$.

\lemref{covered-crossings}, in our new setting, requires some further discussion. Recall that the analysis in its proof 
takes place with respect to a random sample, call it $K_a$, of $n/k$ surfaces of $\Sigma_a^+$. Rather than doing
it for each $a$ separately, we take a single `global' random sample $K$ of $n/k$ surfaces of $\F$.
Consider some specific $z$-vertical covered edge-crossing $(e,e',g)$, with respect to some function $a\in\F$
and its associated structures $\Sigma_a^+$ and $\C_a$, with a suitable surface $\zeta$ 
that crosses $g$.
Then there are five additional functions $f_e^1$, $f_e^2$, $f_{e'}^3$, $f_{e'}^4$, $f_\zeta$ in $\F$, such that
$e$ is a portion of $a\cap f_e^1\cap f_e^2$, $e'$ is a portion of $a\cap f_{e'}^3\cap f_{e'}^4$, and
$\zeta$ is the surface $a\cap f_\zeta$. Then, for $(e,e',g)$ and $\zeta$ to show up in the sample $K$, as a level-1 
vertical edge-crossing, with $\zeta$ as the only `interfering' surface, we need to choose the six functions 
$a$, $f_e^1$, $f_e^2$, $f_{e'}^3$, $f_{e'}^4$, and $f_\zeta$ in $K$, and not choose any of the other $O(k)$ surfaces
whose intersections with $a$ either cross $g$, or cross the portion of the bottom surface, within $V^{e'}$, between 
$g$ and its intersection with $\zeta$ 
(see Figure~\ref{fig:path-2.6}).
The probability for this to happen is at least $c/k^6$, for some suitable constant $c$.
Another useful feature, which we use in the analysis, is that the region $\C_a$ expands in the sample $K$ (i.e.,
the region $\C_a$ in the sample, denoted $\C_a(K)$, contains the region $\C_a$ in the full set $\F$, 
denoted $\C_a(\F)$; this is the monotonicity property observed in Section~\ref{sec:union}), 
provided that $a$ is chosen in $K$. This guarantees that $(e,e',g)$ with $\zeta$ 
(where $\zeta$ is the only `interfering' surface in the sample) does show up as a level-1 
vertical edge-crossing in $\C_a(K)$, once the relevant six surfaces have been chosen in $K$ 
(and the other interfering surfaces have not been chosen).

The proof of the property then proceeds as before, using path drawing and a planarity argument, to bound the number 
of these covered vertical edge-crossings of level 1 in $\C_a(K)$. As before, the bound is linear in the 
complexity $|\C_a(K)|$. Summing this over all the surfaces $a$ that have been chosen in $K$, and taking the
expectation of this sum (with respect to the choice of $K$, using Chernoff's bound to get rid of large deviations),
we get an overall bound, for $K$, of $O^*((n/k)^3)$.
Multiplying by $k^6$, to move from $K$ to the entire collection $\F$, (and replacing $n/k$ by $n$
\cite{CS89}) and by another factor of $k$, the same factor that arises in the original analysis 
in \cite{SS-97} and has been briefly explained in Section~\ref{sec:union}, we get an overall bound of $O^*(k^4n^3)$.

The final inequalities and recurrences are now
\[
|\CSet| = O^*\left( k^4n^3 \right) + O\left( \frac{\beta^2(n)}{k} \visib_{4k}(n) \right) ,
\]
and, from~(\ref{eq:R4d}),
\[
\visib_0(n) \le \frac{\beta(n)}{k} |\CSet| + O^*\left( k^3 n^3\beta^2(n)\right) ,
\]
so
\[
\visib_0(n) = O^*\left(k^3n^3\right) + O\left(k^3\beta^3(n) \visib_0(n/k)\right) ,
\]
which follows from the Clarkson-Shor bound $\visib_{4k}(n) = O\left( k^5 \visib_0\left(\frac{n}{k} \right)\right)$, because 
one needs five functions to define a vertical visibility.
With a suitable choice of $k$, and arguing similarly to solutions of earlier recurrences, 
the final recurrence solves to $\visib_0(n) = O^*(n^3)$, which completes the proof of Theorem~\ref{thm:env4d}.

%------------------------------------------------
\subsection{Entire arrangement}
\label{subsec:4darr}

Let $\S$ be a collection of $n$ semi-algebraic surfaces in $\reals^4$ of constant complexity. For simplicity, assume 
that the surfaces of $\S$ are graphs of totally defined continuous semi-algebraic functions of constant 
complexity, and that they are in general position (with some additional effort, these assumptions can 
be dropped, as above). Koltun~\cite{Koltun-04a} has shown that the complexity of the vertical decomposition 
of $\A(\S)$ is $O^*(n^4)$, but his proof is somewhat involved. In this subsection we give a 
considerably simpler proof, which is based on the machinery developed earlier in the paper.

The basic approach to analyzing the complexity of vertical decompositions in higher dimension uses the following
dimension-reducing strategy, which we review (and adapt) only for $d=4$ dimensions. We fix a pair of surfaces $\sigma_1$,
$\sigma_2\in\S$, and form the intersection surfaces $\sigma\cap\sigma_1$ and $\sigma\cap\sigma_2$, 
for all $\sigma\in\S\setminus\{\sigma_1,\sigma_2\}$, including also $\sigma_1\cap\sigma_2$. We then project these $2n-1$
two-dimensional surfaces onto the $xyz$-space, denote the resulting set of surfaces as $\S_{\sigma_1,\sigma_2}$,
and form the three-dimensional vertical decomposition of their arrangement $\A(\S_{\sigma_1,\sigma_2})$.
We then lift each of the prisms of this decomposition to a four-dimensional prism in $\reals^4$, and clip it to its portion between $\sigma_1$ and $\sigma_2$, keeping only the prisms that are not crossed by any other surface of $\S$.

In a na\"ive implementation of this technique, the number of resulting prisms is $O(n^2)$ (the number of pairs 
$(\sigma_1,\sigma_2)$) times $O^*(n^3)$ (the complexity of each three-dimensional decomposition), namely $O^*(n^5)$.
We improve this bound to $O^*(n^4)$, as follows.\footnote{%
  The weakness of the na\"ive approach is that it does not count only the `empty' lifted prisms,
  namely those not crossed by any surface. Our improved analysis does take this issue into account.}

For each surface $\sigma\in\S\setminus\{\sigma_1,\sigma_2\}$, put 
\[
K_\sigma := \{(x,y,z)\in\reals^3 \mid \sigma_1(x,y,z)\le \sigma(x,y,z) \le \sigma_2(x,y,z) \}.
\]
(Here we only consider the portion of 3-space over which $\sigma_1$ lies below $\sigma_2$.)
Let $\Sigma_{\sigma_1, \sigma_2} := \{ K_\sigma \mid \sigma \in \S \setminus \{\sigma_1,\sigma_2 \}$.
It suffices to construct the vertical decomposition of the complement $\C_{\sigma_1,\sigma_2}$ of
$\bigcup_{\sigma \in \Sigma_{\sigma_1, \sigma_2} } K_\sigma$,
because any point $(x,y,z)$ in the union has an ``interfering'' surface $\sigma$ between
$\sigma_1(x,y,z)$ and $\sigma_2(x,y,z)$.
We then lift each resulting three-dimensional prism $\tau_0$ to the four-dimensional prism
\[
\tau =  \{(x,y,z,w)\in\reals^4 \mid (x,y,z)\in\tau_0 \mbox{ and } \sigma_1(x,y,z)\le w \le \sigma_2(x,y,z) \}.
\]
By construction, the interior of each such prism $\tau$ is not crossed by any surface in $\S$.
These prisms, over all (ordered) choices of $\sigma_1$, $\sigma_2$, are pairwise openly disjoint,
and their union is the entire 4-space. To bound the number of prisms, we proceed as follows.

It suffices to bound the overall number of three-dimensional prisms in the various projections into 3-space.
The analysis in the preceding sections suggests that we begin by bounding the overall complexity of the 
(undecomposed) regions $\C_{\sigma_1,\sigma_2}$ (the complements of the respective unions), 
over all pairs of surfaces $\sigma_1$, $\sigma_2\in\S$.

To bound this complexity, we count the number of vertices in the projected arrangements. Fix $\sigma_1$, $\sigma_2$
as before. Call each projected intersection 2-surface red (resp., blue) if it is the projection of an intersection 
of some surface $\sigma$ with $\sigma_1$ (resp., $\sigma_2$). A vertex of $\A(\S_{\sigma_1,\sigma_2})$ is an intersection 
of either three red surfaces, or three blue surfaces, or two red surfaces and one blue surface, or two blue 
surfaces and one red surface; we refer to these vertices respectively as RRR-vertices, BBB-vertices, RRB-vertices,
and RBB-vertices.

The number of RRR-vertices and BBB-vertices is easy to bound: Each RRR-vertex, say, is a vertex of the three-dimensional
arrangement, within $\sigma_1$, of the intersection 2-surfaces of the form $\sigma_1\cap\sigma$. Their number is therefore
$O(n^3)$, for each $\sigma_1$. By the same argument, the number of BBB-vertices is also $O(n^3)$, for each $\sigma_2$.
Summing over all surfaces, the number of these vertices is $O(n^4)$. 

It is important to notice the following property. Each such vertex $v$, say an RRR-vertex, is independent of $\sigma_2$, 
but, fortunately, $v$ can arise as an RRR-vertex for at most one surface $\sigma_2$. This is because $v$ is a vertex of
$\C_{\sigma_1,\sigma_2}$, which means that the upward $w$-vertical segment erected from $v$ to $\sigma_2$ does not 
intersect any surface of $\S$ before it reaches $\sigma_2$. Hence $v$ arises only once as an RRR-vertex (when
we pair $\sigma_1$ with $\sigma_2$). Similarly, it can arise  at most once as a BBB-vertex. That is, these vertices
contribute a total of $O(n^4)$ to the overall complexity of the regions $\C_{\sigma_1,\sigma_2}$.

Consider then RRB-vertices. Let $v$ be such a vertex, formed by the intersection of two red surfaces, which are
projections of two respective 2-surfaces $\sigma_1\cap\sigma$ and $\sigma_1\cap\sigma'$, and of one blue surface, which 
is the projection of a 2-surface $\sigma_2\cap\sigma''$. Interpreting this back in 4-space, $v$ corresponds to a
vertical visibility between the curve $\gamma = \sigma_1\cap\sigma\cap\sigma'$ and the 2-surface
$\sigma_2\cap\sigma''$. The number of such vertical visibilities, for each curve $\gamma$, is $O^*(n)$, as 
each of them corresponds to a vertex of the lower envelope of the cross-sections of the surfaces of $\S$ within the 
upward vertical curtain erected from $\gamma$. There are in total $O(n^3)$ such curves $\gamma$ (over all choices 
of $\sigma_1$), so the overall number of these vertical visibilities is $O^*(n^4)$. By a fully symmetric argument,
the overall number of RBB-vertices is also $O^*(n^4)$.

That is, the overall complexity of the (undecomposed) regions $\C_{\sigma_1,\sigma_2}$ 
(the complements of the respective unions), over all pairs of surfaces $\sigma_1$, $\sigma_2\in\S$, is $O^*(n^4)$.

We now apply the analysis of the previous sections in this context.
As in the case of lower envelopes in $\reals^4$, reviewed in \subsecref{lower_env}, 
we do not have any sharp bound on the complexity of individual complements $\C_{\sigma_1,\sigma_2}$, 
but only on their overall complexity. We therefore proceed as before, using the adapted machinery of 
\cite{SS-97} (and Section~\ref{sec:union}). This technique charges
vertical visibilities in $\C_{\sigma_1,\sigma_2}$, in the corresponding three-dimensional arrangements
$\A(\S_{\sigma_1,\sigma_2})$, to features of the undecomposed complements. It is also based on random sampling, in which
each 2-surface of the form $\sigma_1\cap\sigma$ or $\sigma_2\cap\sigma'$ is chosen with probability $1/k$, 
for a suitable parameter $k$. To handle these samplings, we take, as in the case of lower envelopes in $\reals^4$,
one global random sample $\R\subset\S$, where each surface is chosen with probability $1/k$. The expected overall complexity
of the undecomposed complements $\C_{\sigma_1,\sigma_2}$, over all pairs $\sigma_1$, $\sigma_2$ that are 
both chosen in $\R$, is $O^*(n^4/k^4)$. Indeed, each vertex of any complement $\C_{\sigma_1,\sigma_2}$ 
is defined by at most four or five surfaces, where there are four defining surfaces when all of them
involve intersections with $\sigma_1$ or all with $\sigma_2$, and five when one surface intersects $\sigma_1$
and two intersect $\sigma_2$ or the other way around. In the former case the bound $O(n^4)$ is immediate,
and in the latter case (five defining surfaces) we argue using lower or upper envelopes within the two-dimensional
curtain erected from a curve along $\sigma_1$ or $\sigma_2$, as appropriate;
we also need to apply Chernoff's bound to get rid of large deviations.

Similarly to the analysis in \subsecref{lower_env} and the derivation of~(\ref{eq:Ra4d}),
we derive the following recurrence:
\begin{equation}
  \label{eq:vcross_arr4d_ab}
  \visib_0(\Sigma_{\sigma_1, \sigma_2})
  \le  \frac{c \beta(n)}{k} |\CSet_{\sigma_1, \sigma_2}| +
  O^*\Bigl( k\beta(k) \left( |\C_{\sigma_1,\sigma_2}| + k\beta(n/k)n^2 
  \right) \Bigr) ,
\end{equation}
where $\visib_0(\Sigma_{\sigma_1, \sigma_2})$ is the number of $z$-vertical visibilities in $\C_{\sigma_1,\sigma_2}$,
$c > 0$ is an appropriate constant, and $\CSet_{\sigma_1, \sigma_2}$ is a set of $z$-vertical edge-crossings,
with respect to $\C_{\sigma_1,\sigma_2}$ (constructed as in Section~\ref{sec:union}). 
Then, summing up over all pairs $\sigma_1, \sigma_2 \in \S$, we obtain (similarly to~(\ref{eq:R4d})):
\begin{equation}
  \label{eq:vcross_arr4d}
  \begin{aligned}
  \visib_0(\S)  & \le
  \frac{c\beta(n)}{k} |\CSet|  + O^*\left( k\beta(k) \left( |\C| + k\beta(n/k)n^4 
  \right) \right) \\
  & =
  \frac{c\beta(n)}{k} |\CSet| 
  + O^*\left( k^2 \beta^2(n) n^4 \right) ,
  \end{aligned}
\end{equation} 
where $\visib_0(\S)$ is the overall number of $z$-vertical visibilities in $\C = \bigcup_{\sigma_1, \sigma_2 \in \S}
\C_{\sigma_1,\sigma_2}$, and $\CSet = \bigcup_{\sigma_1, \sigma_2 \in \S} \CSet_{\sigma_1, \sigma_2}$. 
As argued above, $|\C| = O^*(n^4)$.

As in \subsecref{lower_env}, we need to apply Lemmas~\ref{lem:uncovered-crossings} and~\ref{lem:covered-crossings}
in order to bound $|\CSet|$. In particular, we need to construct a global random sample $\R\subset\S$, 
as in \subsecref{lower_env}, and argue that the expected overall complexity of the undecomposed complements is 
$O^*(n^4/k^4)$. In order to move from $\R$
to the entire collection $\S$, we need to multiply by a factor of $k^6$, and by another factor of $k$ (see once again 
\subsecref{lower_env} for these details), so we therefore obtain an overall bound of $O^*(k^3n^4)$. This bound is
integrated into the corresponding term in Lemma~\ref{lem:uncovered-crossings}, from which we obtain:
\[
|\CSet| = O^*\left( k^3n^4 \right) + O\left( \frac{\beta^2(n)}{k} \visib_{4k}(n) \right) ,
\]
where $\visib_{4k}(n)$ is the overall number of vertical edge-crossings in the three-dimensional 
arrangements, which are at level at most $4k$. Using~(\ref{eq:vcross_arr4d}) we obtain:
\[
\visib_0(n) \le
\frac{c\beta(n)}{k} |\CSet| + O^*\left( k^2 \beta^2(n) n^4 \right) =
O^*\left(k^2 n^4  +  \frac{1}{k^2} \visib_{4k}(n) \right)  .
\]
We next use the Clarkson-Shor technique, which bounds this expression by
\[
O^*\left( k^2 n^4  + \frac{1}{k^2} k^6 \visib_0(n/k) \right) =
O^*( k^2n^4 + k^4 \visib_0(n/k) ) , 
\]
because each vertical visibility (or vertical edge-crossing) is defined by up to six surfaces, namely the two surfaces 
$\sigma_1$, $\sigma_2$ that define the three-dimensional projected arrangement, and by four additional 2-surfaces
in that arrangement, each one of the forms $\sigma_1\cap\sigma$ or $\sigma_2\cap\sigma'$, for
$\sigma$, $\sigma'\in\S$. That is, we obtain a recurrence that, with a suitable sufficiently large value of $k$, 
solves to $O^*(n^4)$, as is easily verified by induction on $n$ (similarly to treatment of earlier recurrences).
We have thus re-established Koltun's bound~\cite{Koltun-04a}, by a considerably simpler argument, based on 
a suitable adaptation of  the machinery developed earlier in this paper. 

\paragraph{Remark.}
It would be tempting to try to extend this approach to vertical decompositions in five dimensions. However, 
even if everything else goes well (a big if), the final recurrence would have the leading term
\[
O^*\left( \frac{1}{k^2} \visib_{4k}(n) \right) =
O^*\left( \frac{1}{k^2} k^8 \visib_0(n/k) \right) =
O^*( k^6 \visib_0(n/k) ) , 
\]
because now a vertical visibility is defined by up to eight surfaces. With such a leading term, the solution would 
have been $O^*(n^6)$, which does not improve the already known bound in five dimensions. Improving this bound is thus still
a major open problem.

%--------------------------------------------------
\section{Constructing Vertical Decompositions and Cuttings, and Point-Enclosure Reporting Queries}
\label{sec:pt-encl}

In this section we first describe an algorithm for computing the vertical decomposition 
and then use it for computing $(1/r)$-cuttings and for constructing a data structure to 
answer point-enclosure reporting queries. For the sake of concreteness, we focus on computing the 
free space (the complement of the union) of a set of geometric (constant-complexity semi-algebraic)
objects in $\reals^3$. The same approach extends to the other cases studied above---see below.

%----------------------------------------------
\subsection{Computing vertical decompositions}
\label{subsec:algo}

Let $\S$ be a collection of $n$ semi-algebraic sets of constant complexity in $\reals^3$.
For any subset $\S' \subseteq \S$ of size $m \le n$, let
$\C(\S')$ denote the complement of the union of $\S'$, $\VeD{\C}(\S')$ its vertical decomposition,
$\psi(m)$ the maximum complexity of $\C(\S')$, and $\VeD{\psi}(m)$ the 
maximum complexity of $\VeD{\C}(\S')$ (over all subsets $\S'$ of size $\le m$). 
By \thmref{main}, $\VeD{\psi}(m) = O^*(m^2+\psi(m))$.
We present a randomized algorithm that constructs $\VeD{\C}(\S)$ in $O^*(n^2+\psi(n))$ 
expected time.
More precisely, it constructs the set of pseudo-prisms (prisms for short) in $\VeD{\C}(\S)$ and its adjacency graph 
$\adj(\S)$ whose vertex set is the set of prisms in $\VeD{\C}(\S)$ and $(\cell,\cell')$ 
is an edge in $\adj$ if $\cell$ and $\cell'$ have overlapping vertical walls.

We use the lazy randomized incremental construction (RIC) of de~Berg~\etal~\cite{BDS}. 
Let $\langle \sigma_1, \sigma_2, \ldots, \sigma_n \rangle$ be a random permutation of $\S$, 
and let $\S_i$ be the prefix of length $i$ of this permutation. We add the sets of $\S$ one 
by one in this order and maintain $\C(\S_i)$, the complement of the union after $i$ steps, 
and $\VeD{\C}(\S_i)$, its vertical decomposition,
in a lazy manner, as in \cite{BDS}. Concretely, suppose we have already computed $\VeD{\C}(\S_{i-1})$. 
To compute $\VeD{\C}(\S_i)$, we insert $\sigma_i$ into $\VeD{\C}(\S_{i-1})$, which requires 
splitting the cells of $\VeD{\C}(\S_{i-1})$ that $\bd\sigma_i$ crosses, removing all cells 
of the resulting subdivision that lie inside $\sigma_i$, and merging some of the other
exterior cells; see, e.g.,~\cite{BDS,SA}. Performing the last two tasks at each 
step may be expensive, so the lazy RIC only splits the cells at each step but performs 
the \emph{clean-up} task, consisting of the last two steps, only at certain steps.
Roughly speaking, after $i$ steps, we maintain a subdivision $\Pi_i$ of prisms,
such that $\bigcup\Pi_i$ contains $\C(\S_i)$ and 
each prism in $\Pi_i$ lying in $\C(\S_i)$ is contained in a prism of $\VeD{\C}(\S_i)$.
In step $i$, we find all cells (prisms)
of $\Pi_{i-1}$ that $\bd\sigma_i$ crosses, split every such prism by $\bd\sigma_i$ into $O(1)$
subcells, and partition the resulting subcells into prisms (all this takes $O(1)$ time and space 
per prism). Let $\Pi'_i$ be the resulting subdivision. If $i<n$ and $i$ is not a power of $2$,
we set $\Pi_i=\Pi'_i$ and proceed to the next step. Otherwise (i.e., $i=2^k$ for some 
integer $k$ or $i=n$), we also perform the following
clean-up task: we identify all cells (prisms) of $\Pi_i$ that lie in $\U(\S_i)$,
mark them \emph{inactive}, and remove them from $\Pi_i$ so that they are not refined any further. 
Next, we merge all cells that lie in one prism of 
$\VeD{\C}(\S_i)$ into a single prism, and set $\Pi_i$ to be the resulting subdivision, which 
is $\VeD{\C}(\S_i)$. To perform these operations efficiently, we maintain two auxiliary structures
(which are standard in randomized incremental constructions):
\begin{enumerate}[(i)] 

\item {\emph{History DAG}:} We maintain a directed acyclic graph $\HDAG$ whose vertex set
is the set of all prisms created 
by the algorithm so far, and $\cell \rightarrow \cell'$ is an edge of $\HDAG$ 
if $\cell'$ was created by splitting $\cell$ (during an insertion step) or by merging $\cell$ 
with some other cells. The cells marked \emph{inactive} become (inactive) leaves of $\HDAG$.
Each node of $\HDAG$ has $O(1)$ out-degree, although its in-degree may be unbounded (because 
of the merging of subregions into a common region). It follows from well-known results
that the expected depth of $\HDAG$ is $O(\log n)$~\cite{BDS,SA}. 

\item {\emph{Adjacency graph}:} We maintain an undirected graph $\adj$ whose vertex set 
is the set of active leaves of $\HDAG$, and $(\cell,\cell')$ is an edge in $\adj$ if 
$\cell$ and $\cell'$ have overlapping  vertical walls. If the clean-up task is performed 
at step $i$, then $\adj$ is the adjacency graph of $\VeD{\C}(\S_i)$ at the end of step $i$. 
\end{enumerate}

By tracing $\bd\sigma_i$ through $\HDAG$, we identify the cells of $\Pi_{i-1}$ that are crossed by 
$\bd\sigma_i$, as well as the cells of $\Pi_{i-1}$ that lie inside $\sigma_i$. 
By traversing $\adj$, starting from the cells of $\Pi'_i$ that were created in Step~$i$
(i.e., the children of cells of $\Pi_{i-1}$ that are crossed by $\bd\sigma_i$),
we can identify all cells of $\Pi_i$ that need to be merged into a single prism. 
We omit the details from here, which are studied, and can be found in~\cite{BDS}.

Set $u =\floor{\log_2 n}$. For $k \in [0, u]$, let $\R_k = \S_{2^k}$, i.e., the subset of $\S$ 
that already has been inserted when the $k$-th clean-up step was performed. 
Set $\R_{u+1} = \S$. For a prism $\cell\in \VeD{\C}(\R_k)$, let 
$\S_{\cell,k} \subseteq \R_{k+1}\setminus\R_k$ 
(resp., $\S_\cell \subseteq \S\setminus\R_k$) be the set of regions of $\R_{k+1}$ 
(resp., $\S$) that intersect $\cell$ (i.e., either their boundaries cross $\cell$ or 
they contain $\cell$). As proved in~\cite{BDS}, for any constant $t>0$,
\begin{equation}
\label{eq:conflict}
\EE\left [ \sum_{\cell\in \VeD{\C}(\R_k)} |\S_{\cell}|^t \right ] = O \left ( \VeD{\psi}(2^k) \biggl ( 
\frac{n}{2^k} \biggr )^t \right ), \qquad
\EE\left [ \sum_{\cell\in \VeD{\C}(\R_k)} |\S_{\cell,k}|^t \right ] = O ( \VeD{\psi}(2^k) 2^t ) , 
\end{equation}
where expectation is with respect to the random insertion order and the second equality is 
obtained using the fact that $|\R_{k+1}| \le 2|\R_k|$.
Roughly speaking, on average, each cell of $\VeD{\C}(\R_k)$ is refined into $O(1)$ prisms 
before the next clean-up step, so the expected time spent in splitting the cells during the
steps $[2^k+1, 2^{k+1}]$ and performing the clean-up at step $2^{k+1}$ can all be charged to 
the cells of $\VeD{\C}(\R_k)$ and $\VeD{\C}(\R_{k+1})$. Using the first bound in~\eqref{conflict}, with $t=1$,
and assuming $\VeD{\psi}(n) = \Omega(n^{1+\delta})$ for some constant $\delta>0$, we also conclude the following:
\begin{equation}
	\label{eq:conflict-size}
\sum_{k=1}^{\ceil{\log_2n}} \sum_{\cell\in \VeD{\C}(\R_k)} |\S_{\cell}| = 
	 O(\VeD{\psi}(n)) = O^* ( n^2 + \psi(n)) .
\end{equation}
The analysis in~\cite{BDS}, whose details are omitted here, implies that the 
expected time spent in computing $\VeD{\C}(\S)$ is $O(\VeD{\psi}(n)\log n) = O^*(n^2+\psi(n))$.

The above algorithm extends to computing the vertical decomposition of the entire arrangement 
$\A(\S)$ (in $\reals^3$). Using the fact that the expected number of vertices in $\A(\R)$,
for a subset $\R\subseteq \S$ of size at most $m$, is $O(m^2+ \chi m^3/n^3)$, where $\chi$ 
is the number of vertices of $\A(\S)$, a similar analysis implies that the expected time 
spent in computing $\VeD{\A}(\S)$ is $O^*(n^2+\chi)$. Finally, the approach
also extends to computing the vertical decomposition of the 
minimization diagram (or the region lying below the lower envelope) of a set 
$\F$ of semi-algebraic trivariate functions of constant complexity in $\reals^4$. 
Since the complexity of the vertical decomposition of the minimization diagram is near cubic, 
the expected time spent in computing $\VeD{M}_\F$, the vertical decomposition of the minimization 
diagram $M_\F$, is $O^*(n^3)$. Omitting further details, we conclude the following:
%----------------------------------------
\begin{theorem}
  \label{thm:algo}
  (i) Let $\S$ be a collection of $n$ constant-complexity semi-algebraic sets in $\reals^3$ such that 
  the complexity of the union of any subset of $\S$ of size at most $m$ is at most $\psi(m)$, for 
  some superlinear function $\psi$.
  Then the vertical decomposition of $\C(\S)$ %(resp., $\A(\S)$)
  can be constructed in $O^*(n^2+\psi(n))$ 
  %(resp., $O^*(n^2+\chi)$)
  randomized expected time.

\medskip\noindent
  (ii) Let $\S$ be as in (i), and let $\chi$ be the number of vertices in $\A(\S)$.
  Then the vertical decomposition of $\A(\S)$ can be constructed in 
  $O^*(n^2+\chi)$ randomized expected time.

\medskip\noindent
(iii) Let $\F$ be a collection of $n$ trivariate semi-algebraic functions of constant complexity.
Then the vertical decomposition of the minimization diagram of $\F$, or of the region in $\reals^4$
below the lower envelope of $\F$, can be constructed in randomized expected time $O^*(n^3)$.
\end{theorem}
%---------------------------------------

%---------------------------------------------------
\subsection{Cuttings}
\label{subsec:cut}

Let $\S$ be a set of $n$ semi-algebraic sets in $\reals^3$ as above.
For a parameter $1<r\le n$, a \emph{$(1/r)$-cutting} of $\C(\S)$
is a set $\Xi$ of pseudo-prisms (again, prisms for short) with pairwise-disjoint relative interiors that cover $\C(\S)$,
such that at most $n/r$ sets of $\S$ \emph{properly intersect}
the relative interior of each prism $\tau \in \Xi$, i.e., the number of regions 
whose interiors contain the relative interior of $\tau$ or whose boundaries cross the 
relative interior of $\tau$ is at most $n/r$. The subset of $\S$ properly intersecting 
$\cell$ is called the \emph{conflict list} of $\cell$.
%\footnote{We emphasize that the conflict list may also include sets that fully contain $\cell$, which is somewhat non-standard in the theory of $(1/r)$-cuttings---see below for the construction.} \esther{This is a new footnote, pls check.}
This definition of cuttings extends to the entire arrangement of a set $\S$ of surface patches in $\reals^3$ by treating 
each surface patch as an infinitely thin semi-algebraic set, as described 
in~\subsecref{output_sensitive}. Finally, it also extends to the region lying below the 
lower envelope of a family $\F = \{ f_1, \ldots, f_n\}$  of trivariate functions, by defining,
for each $i$, the semi-algebraic set $f_i^+ = \{ (x,y) \in \reals^3 \times \reals \mid y\ge f_i(x) \}$,
and then the complement of the union of these semi-algebraic sets is the region below the lower envelope of $\F$.

It is well known that the random-sampling paradigm can be used to construct 
a $(1/r)$-cutting of $\C(\S)$~\cite{Ag:rs,dBS-95,HW87,m-ept-92}. Namely, set $s=cr\log r$, where 
$c$ is a sufficiently large constant. Let $\R \subseteq \S$ be a random subset of $\S$ of size $s$,
and let $\VeD{\C}(\R)$ be the vertical decomposition of $\C(\R)$. For each cell $\cell\in\VeD{\A}(\R)$, let 
$\S_\cell \subset \S$ be the subset of $\S$ that properly intersect $\cell$. 
By construction, $\S_\cell \cap \R = \emptyset$ and $\cell$ is a semi-algebraic set of 
constant complexity. Therefore using a standard random-sampling argument~\cite{Clarkson-87,HW87}, 
it can be shown that $|\S_\cell| \le n/r$ for all $\cell\in\VeD{\A}(\R)$ with probability at least $1/2$,
assuming the constant $c$ is chosen sufficiently large. The sets $\S_\cell$ are the
conflict lists of the respective cells $\cell$. Therefore, $\VeD{\C}(\R)$ is
a $(1/r)$-cutting $\Xi$ of $\C(\S)$ with probability at least $1/2$. 
$\VeD{\C}(\R)$ along with the conflict lists can be constructed by running the algorithm described in \subsecref{algo} for $r$ 
steps (and performing the clean-up task at the end of the $r$-th step as well). Following the analysis 
in~\cite{BDS} and using \eqref{conflict-size}, the expected time taken by the first $r$ steps of the algorithm is
$O^*(nr + n\psi(r)/r)$. 
Cuttings for the arrangements of surface patches in $\reals^3$ or for the minimization diagram (or region below the lower envelope) of trivariate functions can be computed in a similar manner.
Omitting further details, we obtain the following:\footnote{%
  It is possible to reduce the size of the cuttings by a polylogarithmic factor,
  using a two-level sampling scheme as described in~\cite{AMS98,dBS-95,cf-dvrsi-90,m-ept-92}. 
  Since we are using the $O^*(\cdot)$ notation, and ignore subpolynomial factors, 
  we derive the above simpler, albeit slightly weaker, construction.}

%------------------------------------
\begin{theorem}
  \label{thm:cuttings}
(i) Let $\S$ be a collection of $n$ constant-complexity semi-algebraic sets in $\reals^3$, such that 
  the complexity of the union of any subset of $\S$ of size $m$ is at most $\psi(m)$, for some
  superlinear function $\psi$. For any parameter $r \in [1,n]$, a $(1/r)$-cutting of $\C(\S)$,
  of size $O^*(r^2 + \psi(r))$, along with the conflict lists of its cells,
 can be computed in expected time $O^*(nr + n\psi(r)/r)$.
  
  \medskip\noindent
(ii) Let $\S$ be as in (i), and let $\chi$ be the number of vertices in $\A(\S)$.
  For any parameter $r \in [1,n]$, a $(1/r)$-cutting of $\A(\S)$,
  of size $O^*(r^2 +  r^3 \chi/n^3)$, along with the conflict lists of its cells,
 can be computed in expected time $O^*(nr +  \chi(r/n)^2)$.
  
  \medskip\noindent
  (iii) Let $\F$ be a collection of $n$ trivariate semi-algebraic functions of constant complexity. 
For any parameter $r \in [1,n]$, a $(1/r)$-cutting of the minimization diagram of $\F$, or
  of the region lying below the lower envelope of $\F$, of size $O^*(r^3)$, along with the 
  conflict lists of its cells, can be computed in expected time $O^*(nr^2)$.
\end{theorem}
%------------------------------------

%-------------------------------------------------------------
\subsection{Point-enclosure reporting queries}
\label{subsec:pt-encl}

Let $\S$ be a set of $n$ semi-algebraic regions of constant complexity in $\reals^3$. The goal is to 
preprocess $\S$ into a data structure so that all regions of $\S$ that contain a query point 
can be reported quickly; we refer to such queries as \emph{point-enclosure reporting} queries. 
Specifically, by adapting Chan's data structure 
for 3D halfspace range searching~\cite{Chan00}, we present a data structure of $O^*(n^2+\psi(n))$ 
size and expected preprocessing cost, where $\psi$ is as above, that can answer a point-enclosure
query in $O(\log n+k)$ expected time, where $k$ is the output size.
%the worst-case query time is $O((1+k)\log n)$. 
The approach can be extended to preprocess a set $\F$ of $n$ trivariate functions into a data structure 
that reports all functions of $\F$ whose graphs lie below a query point in $\reals^4$, where
the storage and expected preprocessing are now $O^*(n^3)$, and the cost of a query is as above.

\paragraph{Data structure.}

Let $\R_0 \subset \R_1 \subset \cdots \subset \R_u$ be a sequence of random subsets of the input set 
$\S$, where $|\R_i| = 2^i$ and $u = \floor{\log_2 \left(\tfrac{n}{\log_2 n}\right )}$
(i.e., $|\R_u| = 2^u \in [\tfrac{n}{2\log_2 n}, \frac{n}{\log_2 n}]$).
Using the algorithm described in \subsecref{algo}, we construct  
$\VeD{\C}(\R_0), \ldots, \VeD{\C}(\R_u)$, along with the conflict lists of the respective prisms,
as well as the overall history DAG $\HDAG$. For each cell $\cell\in\VeD{\C}(\R_i)$, let $\S_\cell$ 
be the conflict list of $\cell$ with respect to the entire set $\S$. 
For technical reasons, detailed below, to boost the probability of correctness we construct 
three copies of the above data structure, i.e., we choose three independent random permutations 
of $\S$, define the sequence of subsets $\R_{i,j}$, for $1 \le i \le 3$ and $0 \le j \le u$, and 
build $\VeD{\C}(\R_{i,j})$ for each subset separately. Let $\Psi_1, \Psi_2, \Psi_3$ denote 
these three data structures. 
If the size of any $\Psi_i$ 
exceeds $c_0(n^2+\psi(n))n^\eps$, where $c_0>0$ is a sufficiently large constant, we 
reconstruct $\Psi_i$.
The total size and expected preprocessing time are $O^*(n^2+\psi(n))$ (cf.\ \thmref{algo}).

\paragraph{Query procedure.}
Let $q\in\reals^3$ be a query point, let $\S_q \subset\S$ be the set of regions that contain $q$, 
and put $k = |\S_q|$. We assume that we know the value of $k$, up to a factor of $2$, say;
otherwise we perform an exponential search to guess the value of $k$. 
For simplicity of exposition, we assume that $k=n/2^\ell$.
For each $i \le 3$, we search the history DAG $\HDAG_i$ of $\Psi_i$ and find the smallest index 
$\nu_i$ such that  $q \not\in \C(\R_{i,\nu_i})$ (i.e., $q$ lies in the corresponding union). 
For $0 \le j < \nu_i$, let 
$\cell_{i,j}$ be the cell of $\VeD{\C}(\R_{i,j})$ that contains $q$, and let 
$\S_{i,j}$ be the conflict list of $\cell_{i,j}$. By definition $\S_q \subseteq \S_{i,j}$. 
It is known~\cite{BDS} that $\EE[|\S_{i,j}|] \le c n/2^j$ for some constant $c>0$.

We now execute the following iterative procedure, starting with $j = \ell$:
if there exists an $i \le 3$ such that $j < \nu_i$ and $|\S_{i,j}| \le c2^{\ell-j}n/2^j$, then we 
 traverse the conflict list $\S_{i,j}$, report all of its regions that contain $q$, and stop.   
Otherwise we decrement the value of $j$ and repeat. 
Note that $|\S_{i,0}| \le  n$ and $2^\ell = n/k \ge 1$, so the algorithm traverses
the conflict list of one of the $\S_{i,j}$'s and reports $\S_q$. 

We now bound the expected query time. We note that $q \not\in \C(\R_{i,j})$ iff $\S_q\cap\R_{i,j} \ne\emptyset$.
Therefore 
\begin{equation}
\Pr[q \not\in  \C(\R_{i,j})] \le  k\frac{2^j}{n} = \frac{n}{2^\ell} \cdot \frac{2^j}{n} = 2^{-(\ell-j)} .
\end{equation}
Furthermore, by Markov's inequality,
\begin{equation}
	\Pr \biggl [|\S_{i,j}| > 2^{\ell-j} \frac{cn}{2^j}\biggr ] 
	\le \Pr \biggl[|\S_{i,j}| > 2^{\ell-j} \EE[|\S_{i,j}|] \biggr ]\le  \frac{1}{2^{\ell-j}}=2^{-(\ell-j)} .
\end{equation}
Let $\XX_j$ be the event that one of these conditions holds for all three 
data structures (namely, for each $i$ we have $\S_q\cap\R_{i,j} \ne\emptyset$ or 
$|\S_{i,j}| > 2^{\ell - j} \frac{cn}{2^j}$). Then 
\[ 
\Pr[\XX_j] = (2^{-(\ell-j)+1})^3 = 2^{-3(\ell-j)+3} .
\]
If we scan the conflict list for index $j-1$, the time spent is  at most
\[
c 2^{\ell-j+1}\frac{n}{2^{j-1}} = 4c 2^{2(\ell-j)} \frac{n}{2^\ell} = 4c2^{2(\ell-j)} k .
\]
Hence, the expected query time is 
\begin{equation}
	\label{eq:exp-time}
O(\log n) + \sum_{j\leq \ell} \left( 2^{-3(\ell-j)+3} \cdot  4c2^{2(\ell-j)} \right) k = O(\log n+k) .
\end{equation}
%\newtext{
Finally, if the algorithm reports more than $k$ sets, we double the guessed value of $k$ and repeat the procedure. Note that it is only the iterative procedure that we have 
to execute repeatedly for different guesses of $k$, as $\nu_i$ is independent of the value of $k$.
% }
Hence, 
the total expected time in performing the exponential search is 
$O(\log n) + \sum_{i\le \ceil{\log_2 k}} O(2^i) = O(\log n +k)$, 
We thus conclude the following:

%------------------------------------------------
\begin{theorem}
  \label{thm:enc}
  (i) Let $\S$ be a set of $n$ semi-algebraic regions of constant complexity in $\reals^3$.
  $\S$ can be preprocessed, in $O^*(n^2+\psi(n))$ randomized expected time, where $\psi(n)$
  is as above, into a data structure of size $O^*(n^2+\psi(n))$, 
  which supports point-enclosure reporting queries in expected time $O(\log n+k)$, where $k$ is the output size. 

  \smallskip \noindent
  (ii) Let $\F$ be a set of $n$ semi-algebraic trivariate functions of constant complexity. 
  $\F$ can be preprocessed, in $O^*(n^3)$ randomized expected time, 
  into a data structure of size $O^*(n^3)$ that supports point-enclosure reporting queries
  (of the functions that pass below the query point),
  in expected time $O(\log n+k)$, where $k$ is the output size. 
	%The worst-case query time is $O(\log n+k\log(n/k))$.
\end{theorem}
%------------------------------------------------

%-----------------------------------------------------
\section{Conclusion}

In this paper we have settled in the affirmative a few long-standing open problems involving the 
vertical decomposition of various substructures of arrangements in three and four dimensions.
In particular, we have obtained
sharp bounds on the complexity of the vertical decomposition of the complement of the union of a family of 
semi-algebraic sets of constant complexity in $\reals^3$, and of the region below the lower envelope 
of a family of semi-algebraic trivariate functions. We also obtained an output-sensitive bound on the 
size of the vertical decomposition of the arrangement of a family of semi-algebraic sets in $\reals^3$. 
These results lead to efficient algorithms for constructing the vertical decompositions themselves, 
for constructing $(1/r)$-cuttings of %above
substructures of arrangements, and for answering point-enclosure reporting queries.
In a companion paper, 
%Finally, we applied these results to obtain more 
we present efficient data structures for various proximity problems involving lines and points (and segments, and triangles) in $\reals^3$ using the results obtained in this paper.

We conclude by mentioning a few open problems. 
The major long-standing open question is, of course, to improve the bound on the complexity of 
the vertical decomposition of the arrangement of a family of $n$ semi-algebraic sets in $\reals^d$,
for $d\ge 5$, from $O^*(n^{2d-4})$ to $O^*(n^d)$. 
But an immediate open question is whether the techniques developed in this 
paper can be extended to obtain improved bounds on the vertical decomposition of 
various additional substructures of arrangements in $\reals^3$ and $\reals^4$.  
Challenges of this kind include:
\begin{itemize}
	\item[(i)] Show that the complexity of the vertical decomposition of all cells of depth at most $k$ in an 
  arrangement of $n$ semi-algebraic regions in $\reals^3$ is $O^*(n^2+\psi_k(n))$, 
  where $\psi_k(n)$ is the maximum complexity of these (undecomposed) cells.
\item[(ii)] Show that the complexity of the vertical decomposition of the sandwich region lying 
  between the upper envelope of a family of 
  $n$ semi-algebraic trivariate functions (of constant complexity) and the lower envelope 
  of another family of $n$ such functions is $O^*(n^3)$.
  \item[(iii)] Show that the complexity of the vertical decomposition of the minimization diagram of $n$ $4$-variate semi-algebraic functions of
  constant complexity is $O^*(n^4)$.
\end{itemize}

%----------------------------

%---------------------------------------------------------------------------

\end{document}